\documentclass[conference,letterpaper]{IEEEtran}

\usepackage[utf8]{inputenc}
\usepackage[T1]{fontenc}
\usepackage{url}
\usepackage{ifthen}
\usepackage{cite}
\usepackage[cmex10]{amsmath} 
\usepackage{amssymb,amsfonts}
\usepackage{bbm}
\usepackage{graphicx}
\usepackage{subcaption}
\usepackage{placeins}
\usepackage{tikz}
\usepackage{bm}
\usetikzlibrary{arrows.meta}
\usepackage{booktabs}

\newtheorem{theorem}{Theorem}

\newtheorem{corollary}{Corollary}

\newtheorem{remark}{Remark}
\newtheorem{proposition}{Proposition}

\newenvironment{proof}[1][Proof]{\noindent\textbf{#1.} }{\hfill$\blacksquare$\par}

\newcommand{\bDelta}{\bm \Delta}
\newcommand{\bDeltas}{\bm \Delta_s}

\newcommand{\Bi}{B_{\mathbf{i}}}

\newcommand{\ib}{\mathbf{i}}
\newcommand{\Ri}{R_\ib}
\newcommand{\Hi}{H_\ib}
\newcommand{\Pib}{P_\ib}
\newcommand{\gi}{g_\ib}
\newcommand{\Li}{L_\ib}

\newcommand{\bracket}[1]{\left[ #1 \right]}
\newcommand{\curly}[1]{\left\{ #1 \right\}}
\newcommand{\paren}[1]{\left( #1 \right)}
\newcommand{\abs}[1]{\left| #1 \right|}

\newcommand{\R}[1]{\mathbb{R}^{#1}}
\newcommand{\N}[1]{\mathbb{N}^{#1}}

\newcommand{\norm}[2]{\left|\left| #1 \right|\right|_{#2}}

\newcommand{\Esub}[2]{\mathbb{E}_{#2}\left[ #1 \right]}
\newcommand{\E}[1]{\mathbb{E}\!\left[ #1 \right]}

\newcommand{\indic}[1]{\mathbbm{1}_{\left\{#1\right\}}} 

\newcommand{\gauss}[2]{\mathcal{N}\paren{#1, #2}}

\newcommand{\unif}[2]{\mathrm{Unif}\paren{#1, #2}}

\newcommand{\relent}[2]{D\paren{#1||#2}}

 \def\BibTeX{{\rm B\kern-.05em{\sc i\kern-.025em b}\kern-.08em
    T\kern-.1667em\lower.7ex\hbox{E}\kern-.125emX}}
\newcommand{\xv}{\mathbf{x}}

\newcommand{\Ht}{\tilde{H}}
\newcommand{\Xfp}{X_{fp}}

\newcommand{\Xv}{\mathbf{X}}
\DeclareMathOperator*{\argmin}{arg\,min}

\begin{document}
\title{The Entropy of Floating Point Numbers}

\author{%
  \IEEEauthorblockN{Sultan M. Daniels\textsuperscript{*}, Samuel H. D'Ambrosia\textsuperscript{*}, Michael R. DeWeese, and Anant Sahai}
  \IEEEauthorblockA{
                    University of California, Berkeley, CA 94720\\
                    Email: \{shda, sultan\_daniels, deweese, sahai\}@berkeley.edu}
}

\maketitle

\begin{abstract}
    Here we present an analytic approximation for the entropy of floating-point numbers, along with bounds on the error of this approximation. 
    It is well-known that the differential entropy is tightly linked to the discrete entropy of a uniformly quantized random variable. Our approximation uncovers a different quantity that provides this link for floating-point quantization. 
    Additionally, we prove that the entropy of a floating-point quantized random variable is approximately unchanged under scaling. Closed-form expressions for the floating-point entropy of common distributions are provided and compared to exact results. 

    \textsuperscript{*}SMD and SHD contributed equally to this work. This work was supported in part by the H2H8 Organization, the U.S. Army Research Laboratory and the U.S. Army Research Office under Contract No. W911NF-20-1-0151, and NSF AST-2132700. 
\end{abstract}

\section{Introduction}


Floating-point numbers are a central choice for the discrete representation of continuous variables for scientific computing and machine learning applications, since they can have a large dynamic range without sacrificing precision \cite{muller2010handbookfpn, Goldberg1991_FPN_Intro}.
This representation uses a finite register of bits to map a real value to binary scientific notation. 
\begin{table}[htbp]
    \centering
    \renewcommand{\arraystretch}{1.426}
    \begin{tabular}{@{}c@{}}
    \hline\hline
    Approximate entropy: $\tilde H_s(X_{fp}) = p - \tfrac{1}{2} - \Esub{\log(|X|f_X(X))}{f_X\!\!}$ \\
    \hline\hline
    $X \sim \mathcal{N}(0,\sigma^2)$ \\
    $\tilde H_s(X_{fp}) = p + \tfrac{1}{2} \log(2\pi e) + \tfrac{\gamma_e}{2\ln 2} \approx p + 2.46$ \\
    \hline
    $X \sim$ Uniform$(-a,a)$ \\
    $\tilde H_s(X_{fp}) = p +\tfrac{1}{2} + \log e \approx p + 1.94$ \\
    \hline
    $X \sim$ Gamma$(\alpha,\theta)$ \\
    $\tilde H_s(X_{fp}) = p - \tfrac{1}{2} + \alpha\log(e)\paren{1 - \psi(\alpha)} + \log \Gamma(\alpha)$ \\
    \hline
    $X \sim \chi^2(k)$ \\
    $\tilde H_s(X_{fp}) = p - \tfrac{1}{2} + \tfrac{k}{2}\log(e)\paren{1 - \psi\!\paren{\tfrac{k}{2}}} + \log\Gamma\!\paren{\tfrac{k}{2}}$ \\
    \hline
    $X \sim$ Laplace$(0,b)$ \\
    $\tilde H_s(X_{fp}) = p + \tfrac{1}{2} + \tfrac{1+\gamma_e}{\ln 2}\approx p + 2.78$ \\
    \hline
    $X \sim$ Logistic$(0,s)$ \\
    $\tilde H_s(X_{fp}) = p - \tfrac{1}{2} + \tfrac{2 + \gamma_e - \ln(\pi/2)}{\ln 2} \approx p + 2.57$ \\
    \hline
    $X \sim$ Weibull$(\lambda,k)$ \\
    $\tilde H_s(X_{fp}) = p - \tfrac{1}{2} + \tfrac{1+\gamma_e}{\ln 2} - \log k \approx p + 1.78 - \log k$ \\
    \hline
    $X \sim$ Lognormal$(\mu,\sigma^2)$ \\
    $\tilde H_s(X_{fp}) = p - \tfrac{1}{2} + \log\!\paren{\sigma\sqrt{2\pi e}} \approx p + 1.55 + \log \sigma$ \\
    \hline
    $X \sim$ Pareto$(x_m,\alpha)$ \\
    $\tilde H_s(X_{fp}) = p - \tfrac{1}{2} + \log\!\paren{\tfrac{e}{\alpha}} \approx p + 0.94 - \log \alpha$ \\
    \hline
    $X \sim$ Beta$(\alpha,\beta)$ \\
    $\tilde H_s(X_{fp}) = p - \tfrac{1}{2} + \log B(\alpha,\beta)$ \\
    ${}+ \tfrac{(\alpha+\beta-1)\psi(\alpha+\beta) - \alpha\psi(\alpha) - (\beta-1)\psi(\beta)}{\ln 2}$ \\
    \hline
    $X \sim t_\nu(0,s)$ \\
    $\tilde H_s(X_{fp}) = p - \tfrac{1}{2} + \log B\!\paren{\tfrac{\nu}{2},\tfrac{1}{2}}$ \\
    ${}+ \tfrac{\tfrac{\nu+1}{2}\psi\!\paren{\tfrac{\nu+1}{2}} - \tfrac{\nu}{2}\psi\!\paren{\tfrac{\nu}{2}} - \tfrac{1}{2}\psi\!\paren{\tfrac{1}{2}}}{\ln 2}$ \\
    \hline
    $\Xv \sim \gauss{\mathbf{0}_d}{\Sigma}$ \\
    $\tilde H_s(\Xv_{fp}) = d\paren{p + \tfrac{1}{2}\log(2\pi e) + \tfrac{\gamma_e}{2\ln 2}}$ \\
    ${}+ \tfrac{1}{2}\log\!\paren{\tfrac{\prod_{i=1}^{d}\lambda_i(\Sigma)}{\prod_{i=1}^{d}\Sigma_{ii}}}$ \\
    \hline
    \end{tabular}
    \caption{\textit{Closed-form approximations for the floating-point entropy of common distributions.} Here, $p$ is the precision, $\gamma_e$ is the Euler-Mascheroni constant, $\Gamma(\cdot)$ is the gamma function, $\psi(\cdot)$ is the digamma function, $B(\cdot, \cdot)$ is the beta function, $\mathbf{0}_d$ is the $d$-dimensional zero vector, and $\lambda_i(\Sigma)$ is the $i$-th eigenvalue of the positive semidefinite covariance matrix $\Sigma$. For the multivariate Gaussian case, we assume scalar quantization where each component of the random vector is independently quantized. For the univariate distributions, the scale parameter cancels (see Theorem~\ref{thm:scale_invariance}), and in the multivariate case, the entropy is invariant under transformation of $\Xv$ by a diagonal matrix (see Corollary~\ref{cor:vec_scale}).
    }
    \label{tab:closed_form}
\end{table}
As floating-point numbers are heavily used in modern deep learning implementations and hardware \cite{nvidia_nvidia_2025, advanced_micro_devices_inc_amd_2025}, 
there has been renewed interest in developing novel floating-point formats that find the right tradeoffs between precision, dynamic range, memory, computational speed and numerical stability \cite{micikevicius2022fp8formatsdeeplearning, kuzmin_fp8_2022, rouhani_microscaling_2023, googleBFloat16Secret}. Recent work on compressing neural networks empirically studies the entropy of the exponent and mantissa of each parameter \cite{bordawekar2022efloatentropycodedfloatingpoint, hao_neuzip_2024}, since they cannot be losslessly compressed below the lower bound set by their entropy \cite{shannon_communication_1949}. Additionally, physicists are applying thermodynamic bounds to study the energetic costs of computation, where the Shannon entropy of logical states plays a key role \cite{WolpertStochThermo2024, Wolpert2019, Landauer1961, maroney2008physicalbasisgibbsvonneumann, MaroneyGeneralizingLandauer}. Our results can be applied to bound the energetic cost of modeling algorithms with real-valued input data and model parameters quantized onto floating-point numbers \cite{dambrosiaThermodynamicCostsSimple2026}. Also, in communication-constrained control \cite{nairFeedbackControlData2007, kostinaRatelimitedControlSystems2016, kostinaRateCostTradeoffsControl2019, hanQuantizedOnlineLQR2026}, when control inputs must be quantized, there are settings where log-spaced quantizers are optimal for stabilizing linear systems \cite{eliaStabilizationLinearSystems2001, elia2002QuantizedStabilization, okanoMinimumDataRate2014}. As floating-point quantization is similarly log-spaced, analyzing its entropy can aid in the design of hybrid systems with continuous and discrete components.

In the design of variable-rate quantizers, 
the objective is to find the best quantizer that minimizes entropy and distortion \cite{GrayNeuhoff1998Quantization}. It was shown in \cite{gish_asymptotically_1968} that in the high-resolution limit, as the number of quantization levels is taken to infinity, uniform quantizers achieve this optimality. In contrast, we take a fixed nonuniform quantizer, the floating-point representation, and ask a basic question: for a given continuous distribution, what is the output entropy of the floating-point quantizer? 



We present a predictive analytic approximation for the entropy of floating-point numbers, through extending the exact relation between differential entropy and the discrete entropy of uniform quantizers with finite levels given in \cite{kostina_data_2017} to nonuniform quantizers (Section~\ref{sec:rel_diff_disc}), and applying a smoothing approximation for the floating-point bin size (Section~\ref{sec:smooth}). This allows for closed-form expressions for the entropy of common distributions represented on floating-point numbers (Table~\ref{tab:closed_form}, App.~\ref{app:approx_dists}), which can be compared to the numerically evaluated exact discrete entropy (Figs.~\ref{fig:fp_exact_approx_multi_dist} and~\ref{fig:fp_exact_approx}). Bounds on the error of these approximations are provided in Sections~\ref{sec:precise} and~\ref{sec:smooth}. 
Finally, we prove the scale-invariance property of the analog of the differential entropy for floating-point quantization (Section~\ref{sec:scale}), which predicts that the output entropy of a floating-point representation is left unchanged when the underlying continuous random variable is scaled by a constant. 
This contrasts with the scale dependence of differential entropy.

\begin{figure}[htbp]
  \centering
  \begin{subfigure}{0.24\textwidth}
    \centering
    \includegraphics[width=\linewidth]{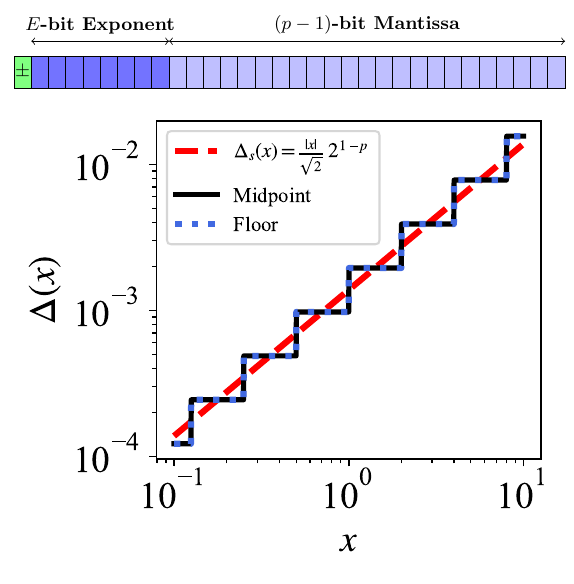}
    \caption{}
    \label{fig:fp_diag}
  \end{subfigure}
    \begin{subfigure}{0.24\textwidth}
    \centering
    \includegraphics[width=\linewidth]{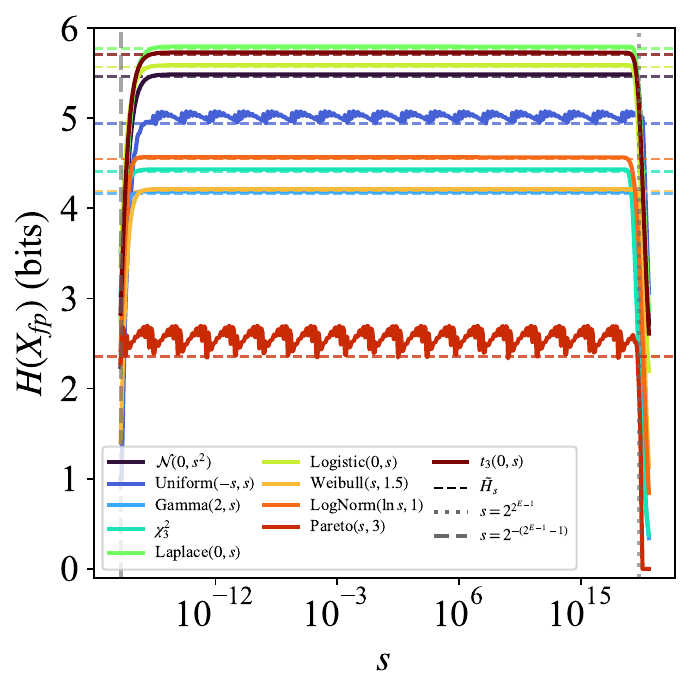}
    \caption{}
    \label{fig:fp_exact_approx_multi_dist}
  \end{subfigure}
\caption{\textit{Floating-point structure and comparison of exact entropy and analytic approximation over the distribution scale parameter}.
(\ref{fig:fp_diag}) The structure of a floating-point number where each box represents one bit. The true bin size function $\Delta(x)$ is plotted on log-log scale for both midpoint (black solid curve) and floor quantization (blue dotted curve) with $p = 10$ and $E=4$ along with the smooth approximation $\Delta_s(x)$ (red dashed curve). 
(\ref{fig:fp_exact_approx_multi_dist}) Comparison between exact entropies and analytic approximations for various distributions (App.~\ref{app:approx_dists}), in terms of a scale parameter $s$ defined in the figure for each distribution, demonstrating excellent agreement inside the underflow and overflow range, $2^{-(2^{E-1}-1)}< s < 2^{2^{E-1}}$. 
}

  \label{fig:fp_struct_and_dists}
\end{figure}

\section{The Structure of a Floating-point Number}
\label{sec:structure}

A normalized floating-point number stores a real value $x \in \R{}$ in the form
\begin{align}
x_{fp}(x) = s_{fp}(x)\times 2^{e_{fp}(x)}\times (1+ m_{fp}(x)),
\end{align}
where $s_{fp}(x)$ corresponds to the sign of the number $x$ encoded by a single bit, $e_{fp}(x)$ is the exponent of the number $x$ encoded using $E$ bits, and $m_{fp}(x)$ is the mantissa or significand specifying the fractional significant digits of $x$ encoded using $p-1$ bits, where $p$ is referred to as the precision of the floating-point number \cite{muller2010handbookfpn, Goldberg1991_FPN_Intro}. Explicit definitions of $s_{fp}(x), e_{fp}(x),$ and $m_{fp}(x)$ are given in Appendix~\ref{app:exact_entropy_float}. This creates a finite set of representable values along the real line. 

The floating-point format can be thought of as a nonuniform quantizer for a continuous random variable $X$. 
The quantization scheme is a map from the real line to $K$ representable values
$\{u_1, u_2, \ldots, u_K\}$, where 
$u_1 < u_2 < \dots < u_K$.
We will refer to the resulting quantized discrete random variable as $X_{fp}$. In this paper, we will consider a simplified floating-point format, with rounding to the nearest representable value, that does not use subnormal numbers (as defined in \cite{muller2010handbookfpn}). The structure of a floating-point number is shown in Fig.~\ref{fig:fp_diag}. See \cite{muller2010handbookfpn, Goldberg1991_FPN_Intro} for further details on the floating-point representation. 

For $E$ exponent bits, $e_{fp}(x)$ is an integer between $e_{min} = -(2^{E-1}-1)$ and $e_{max} = 2^{E-1}$.
Primarily controlled by $e_{max}$, the granular region $\mathbb{U} \triangleq \left[-2^{e_{max} + 1} + 2^{e_{max} - p}, \,\, 2^{e_{max} + 1} - 2^{e_{max} - p}\right]$ is the region of the real line that is near the representable values \cite{GrayNeuhoff1998Quantization}.
Let $\Delta: \mathbb{U} \to \mathbb{R}_+$
encode the bin sizes of the quantization scheme for a single floating-point number.   
Within each exponent interval $[2^{e_{fp}},2^{e_{fp}+1})$, the floating-point format represents $2^{p-1}$ uniformly spaced values. This allows us to define the quantization bin size within an exponent interval as
\begin{equation}
    \label{eq:true_bins}
    \Delta(x) = \frac{2^{e_{fp}+1} - 2^{e_{fp}}}{2^{p-1}}= 2^{e_{fp} - (p-1)} \text{ if } 2^{e_{fp}} \leq x < 2^{e_{fp}+1}.
\end{equation}
Bins at the edges of the granular region and at the boundaries of exponent intervals are handled in Appendix~\ref{app:bin_bounds}. 

We can generalize this to the case of vector quantization, where we assume each component of the vector is independently quantized. Let $\Xv \sim f_\Xv$ be a $d$-dimensional continuous random vector, where $f_\Xv$ is its probability density function. In this case we define $\bDelta : \mathbb{U}^d \rightarrow \R{}_+$, where each quantization bin is a $d$-dimensional rectangle with volume
\begin{equation}
\bDelta(\xv) \triangleq \prod_{j=1}^{d} \Delta(x_j),
\qquad
\log \bDelta(\xv) = \sum_{j=1}^{d} \log \Delta(x_j).
\end{equation}

In this paper, $\log[\cdot]$ is assumed to be base 2. Let $\mathbf{i} = (i_1,\dots,i_d)$ index the $d$-dimensional quantization bins, where $i_j \in \{1,\dots,K\}$. Define the bin boundaries
\begin{equation}
a_i=\begin{cases}
-2^{e_{max} + 1} + 2^{e_{max} - p} & i=1\\[2pt]
\frac{u_i+u_{i-1}}2 & 2\leq i\leq K,
\end{cases}
\end{equation}
\begin{equation}
\qquad
b_i=\begin{cases}
\frac{u_{i+1}+u_i}2 & 1\leq i\leq K-1\\[2pt]
2^{e_{max} + 1} - 2^{e_{max} - p} & i=K,
\end{cases}
\end{equation}
where for each $d$-dimensional index $\mathbf{i}$ we can define the quantization bin
\begin{equation}
B_{\mathbf{i}} \triangleq \prod_{j=1}^d [a_{i_j}, b_{i_j}].
\end{equation}
Let
\begin{equation}
p_{\mathbf{i}} \triangleq \mathbb{P}\{\Xv\in B_{\mathbf{i}}\}
= \mathbb{P}\{\Xv_{fp} = \mathbf{u}_{\mathbf{i}}\},
\end{equation}
where $\mathbf{u}_{\mathbf{i}}$ is the $d$-dimensional representable value corresponding to the quantization bin $B_\mathbf{i}$, and $\Xv_{fp}$ is the discrete random vector representing the floating-point quantization of $\Xv$.

Then the discrete entropy of $\Xv_{fp}$ is
\begin{equation}
H(\Xv_{fp}) = -\sum_{\mathbf{i}} p_{\mathbf{i}} \log p_{\mathbf{i}}.\label{eq:exact_entropy}
\end{equation}
See Appendix~\ref{app:exact_entropy_float} for a more detailed discussion of the exact entropy calculation used in Figures~\ref{fig:fp_exact_approx_multi_dist} and~\ref{fig:fp_exact_approx}.

\begin{table*}[!t]
    \centering
    \renewcommand{\arraystretch}{2}
    \begin{tabular}{@{}cccc@{}}
    \hline\hline
    & Kostina \cite{kostina_data_2017} & This paper & Equation\\
    \hline\hline
     Differential $\leftrightarrow$ discrete & $h(\Xv) - \log \bDelta + \relent{f_\Xv}{g}$ & $h(\Xv) - \Esub{\log \bDelta(\Xv)}{f_\Xv} + \relent{f_\Xv}{g}$ & Eq.~\eqref{eq:exact_identity_main}\\
    \hline
    $\relent{f_\Xv}{g} \leq $ &$\bDelta\cdot\Esub{\max\limits_{u \in B_\ib(\Xv)}\norm{\nabla \ln f_\Xv(u)}{2}}{f_\Xv}$ & $\sum\limits_{\ib}p_\ib
    \int_1^{\Lambda_\ib}
    \paren{\log \lambda+\log e}
    \Li(\lambda)
    \,d\lambda$ & Theorem~\ref{thm:continuum_bounds_main} \\
    \hline
    $\relent{f_\Xv}{g} \geq $ & N/A & $\sum\limits_{\ib}p_\ib \max\limits_{\lambda \geq 1} \bracket{
\lambda\log \lambda
-
(\lambda-1)\log e
}\Li(\lambda)$ & Theorem~\ref{thm:continuum_bounds_main} \\
    \hline
    \end{tabular}
    \caption{\textit{Comparing and contrasting to Kostina \cite{kostina_data_2017}}. We extend the relation of differential entropy to the discrete entropy of the quantized output to the setting of nonuniform bins. $\bDelta$ denotes constant bin size, while $\bDelta(\Xv)$ denotes a bin size that depends on the value of $\Xv$. $\Bi(\Xv)$ denotes the bin that $\Xv$ gets quantized into. The upper bound on $\relent{f_\Xv}{g}$ in \cite{kostina_data_2017} uses the first-order information of the density to provide a smoothness condition that bounds the KL divergence for high-resolution limit analysis. Theorem~\ref{thm:continuum_bounds_main} provides an upper bound, and a similar lower bound, by characterizing the peaks of the density in each bin.}
    \label{tab:kostina}
\end{table*}

\section{Relation of differential entropy to discrete floating-point entropy}\label{sec:rel_diff_disc}

We can extend the analysis of \cite{kostina_data_2017} to the case of quantization with nonuniform bins, to derive an exact identity relating the discrete entropy in \eqref{eq:exact_entropy} to the differential entropy of $\Xv$. Define the piecewise-uniform density
\begin{equation}
g(\xv)
\triangleq
\sum_{\mathbf{i}}
\frac{p_{\mathbf{i}}}{\abs{\Bi}}
\indic{\xv\in B_{\mathbf{i}}},
\label{eq:piecewise_uniform_g_main}
\end{equation}
where $\abs{\Bi} = \bDelta(\xv)$ for $\xv \in B_{\mathbf{i}}$ is the volume of each rectangular bin. By construction, $g$ is constant on each bin and integrates to one across all bins.

By the definition of Kullback-Leibler divergence,
\begin{align}
\relent{f_{\Xv}}{g}
&=
-h(\Xv)
-
\int_{\mathbb{U}^d}
f_{\Xv}(\xv)\log g(\xv)
\,d\xv.
\label{eq:KL_start_main}
\end{align}
Since $g(\xv) = p_{\mathbf{i}}/\abs{\Bi}$ for $\xv\in B_{\mathbf{i}}$,
\begin{align}
\int_{\mathbb{U}^d}
f_{\Xv}(\xv)\log g(\xv)
\,d\xv
&=
\sum_{\mathbf{i}} p_{\mathbf{i}}\log p_{\mathbf{i}}
-
\sum_{\mathbf{i}} p_{\mathbf{i}}\log \abs{\Bi}.
\label{eq:KL_middle_main}
\end{align}
Substituting Eq.~\eqref{eq:KL_middle_main} into Eq.~\eqref{eq:KL_start_main} yields
\begin{align}
\relent{f_{\Xv}}{g}
&=
-h(\Xv)
-
\sum_{\mathbf{i}} p_{\mathbf{i}}\log p_{\mathbf{i}}
+
\sum_{\mathbf{i}} p_{\mathbf{i}}\log \abs{\Bi}
\notag\\
&=
-h(\Xv)
+
H(\Xv_{fp})
+
\mathbb{E}[\log \bDelta(\Xv)],
\end{align}
or equivalently,
\begin{equation}
\boxed{
H(\Xv_{fp})
=
h(\Xv)
-
\mathbb{E}[\log \bDelta(\Xv)]
+
\relent{f_{\Xv}}{g}.
}
\label{eq:exact_identity_main}
\end{equation}

\subsection{A first approximation}
The first approximation made corresponds to neglecting the nonnegative correction term $\relent{f_{\Xv}}{g}$:
\begin{equation}
\boxed{
\Ht(\Xv_{fp})
\triangleq
h(\Xv)
-
\mathbb{E}[\log \bDelta(\Xv)].
}
\label{eq:approx1_clean_main}
\end{equation}

\subsection{Distributions that lead to small approximation error}\label{sec:precise}
Since \eqref{eq:approx1_clean_main} neglects the KL divergence term, 
two questions arise: what properties of $f_{\Xv}$ ensure that this correction term is small, and what properties make this term large?
Since $\relent{f_\Xv}{g}$ captures how far the density $f_\Xv$ is from uniform on each bin, it will be shown that ``spikiness'' is the key property that determines this term, with spikes in small bins less harmful than spikes in large bins. Theorem~\ref{thm:continuum_bounds_main} makes this precise and Remark~\ref{rem:single_peak} provides more geometric intuition.

For $\xv\in \Bi$ with $p_\ib > 0$, let $\gi \triangleq g(\xv)$ and define the relative density ratio by
$
    \Ri(\xv)
    \triangleq
    \frac{f_\Xv(\xv)}{\gi}.
$ 
For $\lambda > 0$, define
\begin{align}
    \Li(\lambda)
    \triangleq
    \frac{1}{\abs{\Bi}}\abs{
        \curly{
            \xv\in \Bi:
            \Ri(\xv)\ge \lambda
        }
    },
\end{align}
which encodes the normalized width of $\Ri$'s $\lambda$-superlevel set. 

As the average over a bin $\gi$ follows the absolute peak density in that bin, if $\gi$ is held constant, a peak of $\Ri(\xv)$ of a certain height, then constrains the peak's width. The following proposition proves that a relative peak of height $\lambda \geq 1$ can occupy at most $1/\lambda$ of the bin width $|\Bi|$.

\begin{proposition}[Peak height-width feasibility bound]\label{lem:feasibility}
For every bin $\Bi$ where $p_\ib > 0$,
\begin{align}
    \Esub{\Ri(\Xv)}{\mathrm{Unif}(\Bi)} = 1.
\end{align}
If $\lambda \geq 1$, this induces the peak height-width feasibility bound
\begin{align}
    \Li(\lambda)
    \le
    \frac{1}{\lambda}.
\end{align}
\end{proposition}
\begin{proof}
    First, observe that
\begin{align}
    &\Esub{\Ri(\Xv)}{\mathrm{Unif}(\Bi)} =
    \frac{1}{\abs{\Bi}}
    \int_{\Bi}\Ri(\xv)\,d\xv\\
    &=
    \frac{1}{\abs{\Bi}}
    \int_{\Bi}\frac{f_\Xv(\xv)}{\gi}\,d\xv
    =
    \frac{1}{\abs{\Bi} \gi}
    \int_{\Bi}f_\Xv(\xv)\,d\xv
    \\
    &=
    \frac{p_\ib}{\abs{\Bi} \gi}
    =
    1.
\end{align}

Let $\pi_\ib(\lambda) \triangleq \curly{\xv \in \Bi: \Ri(\xv)\ge \lambda}$. For $\lambda\ge 1$, by the definition of $\Li$,
\begin{align}
    1
    &=
    \frac{1}{\abs{\Bi}}
    \int_{\Bi}\Ri(\xv)\,d\xv 
    \ge
    \frac{1}{\abs{\Bi}}
    \int_{\pi_\ib(\lambda)}
    \Ri(\xv)\,d\xv
    \\
    &\ge
    \lambda
    \frac{1}{\abs{\Bi}}
    \abs{\pi_\ib(\lambda)}
    =
    \lambda \Li(\lambda).
\end{align}
Therefore
$
    \Li(\lambda)
    \le
    \frac{1}{\lambda}.
$
\end{proof}

Now, intuitively, we upper bound the KL divergence by integrating over relative peak heights $\lambda$ on a per bin basis. 
The lower bound integrates the minimum contribution to the per bin KL divergence of all peaks with relative heights of one up to $\lambda^\star$, over the region containing the peaks of height at least $\lambda^\star$, where $\lambda^\star$ maximizes the tradeoff between peak height and width.
\begin{theorem}[Bounds on the KL divergence]\label{thm:continuum_bounds_main} Under the convention $0\log 0 = 0$, assume
\begin{align}
    \Lambda_\ib
    \triangleq
    \operatorname*{ess\,sup}_{\xv\in \Bi} \Ri(\xv)
    =
    \frac{1}{\gi}
    \operatorname*{ess\,sup}_{\xv\in \Bi} f_\Xv(\xv) < \infty.
\end{align}

Then, 
\begin{align}
    &\sum\limits_{\ib}p_\ib
\max_{\lambda \geq 1} \bracket{
\lambda\log \lambda
-
(\lambda-1)\log e
}
\Li(\lambda)\\
&\leq \relent{f_\Xv}{g}
    \le
    \sum\limits_{\ib}p_\ib
    \int_1^{\Lambda_\ib}
    \paren{\log \lambda+\log e}
    \Li(\lambda)
    \,d\lambda .
\end{align}
\end{theorem}
See Appendix~\ref{app:kl_bounds} for the proof. Applying the peak-width feasibility bound in Proposition~\ref{lem:feasibility} leads to the looser worst-case upper bound $\relent{f_\Xv}{g}
    \le
    \sum\limits_{\ib}p_\ib
    \int_1^{\Lambda_\ib}\frac{1}{\lambda}
    \paren{\log \lambda+\log e}
    \,d\lambda$, that integrates over maximally wide peaks of relative height $\lambda$ until the maximum relative height $\Lambda_\ib$ over each bin.
We can provide more geometric intuition by assuming the density only has one peak in each bin.
\begin{remark}[The one peak per bin height-width upper bound]\label{rem:single_peak} Suppose the above-average region in
$\Bi$ is contained in a measurable set $\Pib \subset \Bi$ with width
$
    \abs{\Pib}
    =
    w_\ib.
$
Let $\Hi \geq \gi$ be the max peak height of the density, so $f_\Xv(\xv)\leq \Hi$ for all $\xv\in \Pib$,

Then, for every $\lambda \ge 1$, the superlevel set $\curly{\xv\in\Bi:\Ri(\xv)\ge\lambda}$ is contained in $\Pib$, so by the definition of $L_i$,
$
    \Li(\lambda)
    \le
    \frac{w_\ib}{\abs{\Bi}}.
$ Then, the upper bound in Theorem~\ref{thm:continuum_bounds_main} gives
\begin{align}
    &\relent{f_\Xv}{g} = \sum\limits_{\ib}\int_{\Bi}f_\Xv(\xv)\log{\frac{f_\Xv}{\gi}}\,d\xv\\
    &\le
    \sum\limits_{\ib}p_\ib \frac{w_\ib}{\abs{\Bi}}
    \int_1^{\Hi/\gi}
    \paren{\log \lambda+\log e}
    \,d\lambda
    \\
    &=
    \sum\limits_{\ib}p_\ib \frac{w_\ib}{\abs{\Bi}}
    \frac{\Hi}{\gi}\log \frac{\Hi}{\gi} =  \sum\limits_{\ib}w_\ib
    \Hi\log \frac{\Hi\abs{\Bi}}{p_\ib}.
\end{align}

This is the one peak per bin height-width bound. The ratio $\frac{\Hi\abs{\Bi}}{p_\ib} \geq 1$ is the area of the rectangle with the height of the peak and the width of the entire bin, divided by the probability mass in that bin. 
For smooth densities and small bins, $\frac{\Hi\abs{\Bi}}{p_\ib} \approx 1$. 
The KL divergence bound increases only logarithmically with this ratio, while for $w_\ib\Hi$, the maximum area under the peak, the KL divergence bound increases linearly. In the setting of floating-point representations, this shows the absolute height $\Hi$ of a peak is not the only important quantity for controlling the KL divergence term. The KL divergence term can still be small for peaked distributions if the peaks have very small width, or if they are near the origin where the bins are very small.

\end{remark}

\section{Smoothing and Extending the Bin Size Function}\label{sec:smooth}

Our second approximation will alter the bin size function given in \eqref{eq:true_bins}. First, we can smooth the steps by introducing a best-fit linear approximation as shown by the dashed red line in Figure~\ref{fig:fp_diag}, which uses $e_{fp}(x) \approx e_{s}(x) \triangleq \log(|x|/\sqrt{2})$. This allows us to define a smoothed bin size function $\Delta_s: \R{} \rightarrow \R{}_+$ by
\begin{equation}
    \label{eq:step_approx}
    \Delta_s(x) \triangleq \frac{1}{\sqrt{2}}|x| \cdot 2^{1-p} \approx \Delta(x).
\end{equation}


Extending the domain of the bin size function from the granular region $\mathbb{U}$ to $\R{}$ allows for initially truncated integrals to then be given in concise closed form expressions. Again, we can generalize to $d$-dimensional vector quantization by defining $\bDeltas(\Xv) \triangleq \prod_{j = 1}^d\Delta_s(X_j)$. Building upon the first approximation $\tilde H(\Xv_{fp})$ given by \eqref{eq:approx1_clean_main}, we define the smoothed and extended approximation
\begin{equation}
    \boxed{
    \begin{aligned}
        \tilde H_s(\Xv_{fp}) &\triangleq h(\Xv) - \mathbb{E}[\log(\bDeltas(\Xv))] \label{eq:approx2_main}\\
        & = d\paren{p - \frac{1}{2}} + h(\Xv) - \sum\limits_{j =1}^d \E{\log|X_j|},
    \end{aligned}
    }
\end{equation}
where we assume the integrals defining the differential entropy and first log-moment of the folded (absolute value) continuous distribution are given by integrals over the whole real line.

For a $d$-dimensional quantization, we prove that the error caused by introducing this smooth approximation is bounded by $d/2$, while the extension error is given simply by the magnitude of the contribution of overflow and underflow values to the entropy integral. For many distributions, the smoothing error will be much smaller than $d/2$ due to cancellation in the overestimation and underestimation of the bin size.

\begin{theorem}[Error bound for smoothing and extending the bin-size function]\label{thm:steps}
Let $\Xv\sim f_\Xv$ be a $d$-dimensional random vector with the probability density $f_\Xv$. Let $\mathbb{O}^d \triangleq \mathbb{R}^d \setminus \mathbb{U}^d$ be the overflow region, and the underflow region be $\mathbb{S}^d \triangleq [-2^{e_{min}}, 2^{e_{min}}]^d$. Then
\begin{equation}
\abs{\Ht(\Xv_{fp}) - \Ht_s(\Xv_{fp})} \leq \frac{d}{2} + \epsilon,
\end{equation}
where 
\begin{equation}
    \epsilon = \abs{
    \int_{\mathbb{S}^d \cup \mathbb{O}^d}
    f_{\Xv}(\xv)
    \log\left[
    f_{\Xv}(\xv)\bDeltas(\xv)
    \right]
    d\xv}.
\end{equation}
\end{theorem}

\begin{proof}[Abridged Proof]

The bound can be shown by considering the ratio between bin sizes. There are three types of bins for which to bound this ratio: bins on the interior of an exponent block, bins on the outer edges, and bins on the boundary between exponent blocks. Here, the bound is given for the interior bins; the other two cases proceed similarly and are given in Appendix~\ref{app:bin_bounds}. The bin size within an exponent block $e$ is defined by Eq.~\eqref{eq:true_bins},
\begin{equation}
    \label{eq:real_bins}
    \Delta(x_j) = 2^{e-(p-1)} \quad \text{for } 2^e \leq |x_j| < 2^{e+1},
\end{equation}
and within this range $\Delta_s(x_j) = \frac{|x_j|}{\sqrt{2}} \cdot 2^{1-p}$ satisfies
\begin{equation}
    \begin{aligned}
    &\frac{2^e}{\sqrt{2}}2^{1-p} \leq \Delta_s(x_j) \leq \frac{2^{e+1}}{\sqrt{2}}2^{1-p}\\
    \implies &\frac{1}{\sqrt{2}} \leq \frac{\Delta_s(x_j)}{\Delta(x_j)} < \sqrt{2}.\label{eq:ratio_bound}
    \end{aligned}
\end{equation}

From~\eqref{eq:approx1_clean_main},~\eqref{eq:approx2_main}, and the triangle inequality,
\begin{align}
\big|
\Ht(\Xv_{fp}) - &\Ht_{s}(\Xv_{fp}) \big| \leq \abs{\sum_{j=1}^{d} \int_{\mathbb{U}\setminus \mathbb{S}}f_{X_j}(x_j)\log\bracket{\frac{\Delta_s(x_j)}{\Delta(x_j)}}\,dx_j}\nonumber\\
&+\abs{\int_{\mathbb{S}^d \cup \mathbb{O}^d}
    f_{\Xv}(\xv)
    \log\left[
    f_{\Xv}(\xv)\bDeltas(\xv)
    \right]
    d\xv}.\label{eq:expect_intermediate}
\end{align}

Applying the bound in~\eqref{eq:ratio_bound} to the integral over $\mathbb{U}\setminus \mathbb{S}$ term in~\eqref{eq:expect_intermediate} completes the proof.
\end{proof}


This implies that the error of the smoothing and extending approximation is strictly bounded by $d/2$ plus the contribution of the overflow and underflow parts of the continuous distribution to be quantized.

\begin{figure}[htbp]
  \centering
  \begin{subfigure}{0.241\textwidth}
    \centering
    \includegraphics[width=\linewidth]{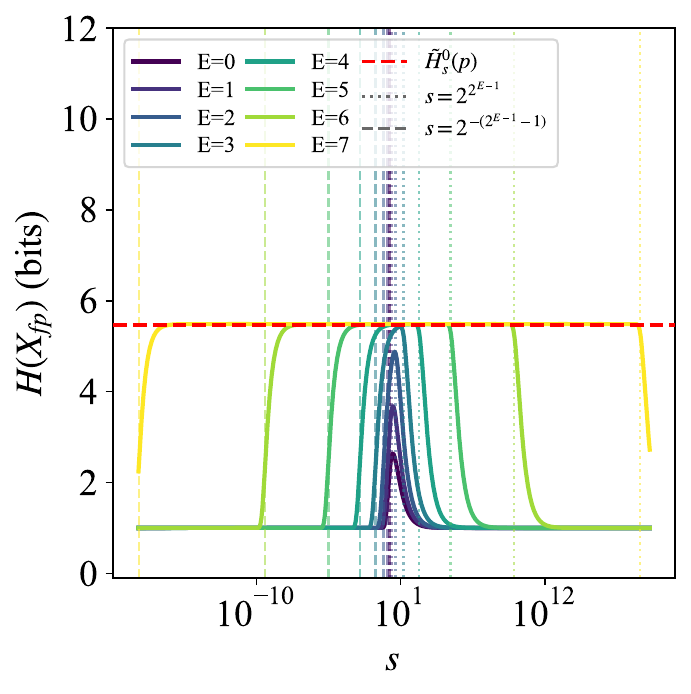}
    \caption{}
    \label{fig:fp_exact_approx_sigma}
  \end{subfigure}
  \begin{subfigure}{0.241\textwidth}
    \centering
    \includegraphics[width=\linewidth]{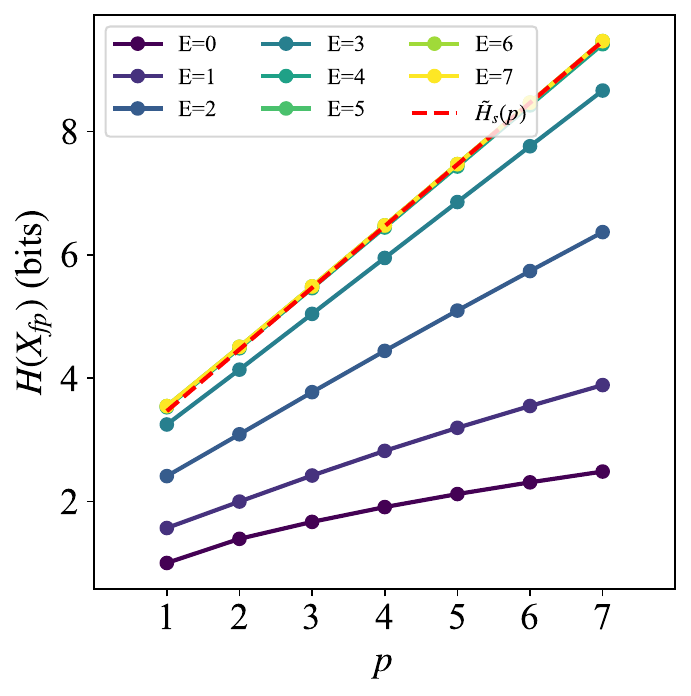}
    \caption{}
    \label{fig:fp_exact_approx_p_E}
  \end{subfigure}
\caption{\textit{Comparisons across exponent and mantissa bit lengths}. 
(\ref{fig:fp_exact_approx_sigma}) shows the entropy of $\Xfp$ which is the clipped and midpoint quantized representation of $X \sim \gauss{0}{s^2}$. Here precision $p=3$, and solid curves correspond to various numbers of exponent bits $E$. The horizontal red line shows the approximate entropy $\Ht_s(p)$. Vertical dashed and dotted lines mark $s = 2^{e_{\min}}$ and $s = 2^{e_{\max}}$, respectively for each $E$. Notice that the exact entropy closely follows the approximate entropy until $s$ approaches the overflow and underflow boundaries for each $E$. 
(\ref{fig:fp_exact_approx_p_E}) shows the exact entropy of $\Xfp$ as $p$ varies for each $E$ when the underlying random variable is $X\sim \gauss{0}{1}$. The approximation $\Ht_s(p)$ is also plotted as the dashed red line. Notice that the curves for $E \geq 4$ are directly on top of each other, even with low $p$, showing that $\Ht_s(p)$ is close to $H(\Xfp)$ whenever $E$ is large enough to keep the probability of overflows and underflows low.
}

  \label{fig:fp_exact_approx}
\end{figure}
\begin{table}[htbp]
    \centering
    \renewcommand{\arraystretch}{2}
    \begin{tabular}{@{}cccc@{}}
    \hline\hline
    & $f_X$ independent & $f_X$ dependent  & Error\\
    \hline\hline
    Float & $p-\frac{1}{2}$ & $-\E{\log\!\paren{|X|f_X(X)}}$ & $\frac{1}{2}+\epsilon+\relent{f_X}{g}$ \\
    \hline
    Uniform & $\log[1/\Delta]$ & $-\E{\log\!\paren{f_X(X)}} $ & $\relent{f_X}{g}$ \\
    \hline
    \end{tabular}
    \caption{\textit{Contributions to the discrete entropy for uniform and floating-point scalar quantization.} The exact entropy for each quantizer is the sum of its row. For both uniform and floating-point quantization, part of the approximation for the discrete entropy is independent of the distribution, and part is dependent. In the uniform case, the distribution dependent contribution is the differential entropy. In the floating-point case, this is $-\mathbb{E}[\log(|X|f_X(X))]$, which we call the \textit{differential floating-point entropy}.
    }
    \label{tab:uni_float_comparison_transposed}
\end{table}

\section{Scale-invariance}\label{sec:scale}
Here, we provide a simple proof that the approximate floating-point entropy does not change when the underlying continuous random variable is scaled by a constant. This predicts that the exact floating-point entropy is approximately scale-invariant as long as the probability of overflow and underflow is low. Figure~\ref{fig:fp_exact_approx_multi_dist} illustrates this scale-invariance until the scale parameter of the underlying distribution approaches the maximum or minimum exponent of the floating-point representation.

\begin{theorem}[Approximate floating-point entropy is scale-invariant for scalar random variables]\label{thm:scale_invariance}
 Let $X \sim f_X$, where $f_X$ is a density, and let $a \in \R{}\setminus \{0\}$. Let $\Xfp$ be the random variable resulting from floating-point quantization of $X$, and let $(aX)_{fp}$ be the random variable resulting from floating-point quantization of $aX$. Then the approximate floating-point entropy of $(aX)_{fp}$ has no dependence on $a$,
 \begin{align}
    \Ht_s((aX)_{fp}) = p - \frac{1}{2} + h(X) - \E{\log\abs{X}}.
 \end{align}
\end{theorem}

\begin{proof}
    From \cite{CoverThomas}, we know that $h(aX) = h(X) + \log|a|$. This means,
   \begin{align}
    &\Ht_s((aX)_{fp}) = p - \frac{1}{2} + h(aX) - \E{\log\abs{aX}}\\
    &= p- \frac{1}{2} + h(X) + \log|a| - \log|a| - \E{\log\abs{X}}\\
    &=p - \frac{1}{2} + h(X) - \E{\log\abs{X}}.
 \end{align}  
\end{proof}

\begin{corollary}[Approximate floating-point entropy for random vectors is invariant to diagonal matrix transformations]\label{cor:vec_scale}
     Let $A = \mathrm{diag}\{a_1, \dots, a_d\}$ be a diagonal matrix with $a_j \neq 0$ for all $j \in \{1,\dots,d\}$, and let $\Xv \sim f_\Xv$ be a $d$-dimensional random vector. Then
 \begin{align}
    \Ht_s((A\Xv)_{fp}) = d\paren{p - \frac{1}{2}} + h(\Xv) - \sum_{j=1}^{d}\E{\log\abs{X_j}}.
 \end{align}
\end{corollary}
\begin{proof}
From \cite{CoverThomas}, $h(A\Xv) = h(\Xv) + \log\abs{\det A}$. The claim then follows by applying Theorem~\ref{thm:scale_invariance} component by component.
\end{proof}

\section{Discussion}\label{sec:disc}

For a sufficiently smooth distribution and sufficiently high range as set by $E$, the entropy of a floating-point quantized random variable can be approximated by
\begin{equation}
\boxed{
H(X_{fp}) \approx \tilde H_s(X_{fp})
\triangleq p-\frac{1}{2} - \mathbb{E}[\log(|X|f_X(X))],
}
\label{eq:approx_discussion}
\end{equation}
where the differential floating-point entropy, $- \mathbb{E}[\log(|X|f_X(X))]$, is the differential entropy $h(X)$ minus the first log-moment of its corresponding folded distribution $\mathbb{E}[\log |X|]$. The error on this approximation is bounded by 
\begin{equation}
\boxed{
|H(X_{fp}) - \tilde H_s(X_{fp})| \leq \frac{1}{2} + \epsilon +\relent{f_{X}}{g}
},
\label{eq:approx_error_discussion}
\end{equation}
where $\epsilon$ is the contribution of the underflow and overflow regions to the entropy, and $\relent{f_X}{g}$ is the Kullback-Leibler divergence between the true distribution and its piecewise constant approximation.

\bibliographystyle{IEEEtran}
\bibliography{mmc_lin_reg}

\begin{thebibliography}{10}
\providecommand{\url}[1]{#1}
\csname url@samestyle\endcsname
\providecommand{\newblock}{\relax}
\providecommand{\bibinfo}[2]{#2}
\providecommand{\BIBentrySTDinterwordspacing}{\spaceskip=0pt\relax}
\providecommand{\BIBentryALTinterwordstretchfactor}{4}
\providecommand{\BIBentryALTinterwordspacing}{\spaceskip=\fontdimen2\font plus
\BIBentryALTinterwordstretchfactor\fontdimen3\font minus \fontdimen4\font\relax}
\providecommand{\BIBforeignlanguage}[2]{{%
\expandafter\ifx\csname l@#1\endcsname\relax
\typeout{** WARNING: IEEEtran.bst: No hyphenation pattern has been}%
\typeout{** loaded for the language `#1'. Using the pattern for}%
\typeout{** the default language instead.}%
\else
\language=\csname l@#1\endcsname
\fi
#2}}
\providecommand{\BIBdecl}{\relax}
\BIBdecl

\bibitem{muller2010handbookfpn}
J.~M. Muller, \emph{\BIBforeignlanguage{eng}{Handbook of floating-point arithmetic / Jean-Michel Muller [and others].}}\hskip 1em plus 0.5em minus 0.4em\relax Boston: Birkhauser, 2010.

\bibitem{Goldberg1991_FPN_Intro}
\BIBentryALTinterwordspacing
D.~Goldberg, ``What every computer scientist should know about floating-point arithmetic,'' \emph{ACM Comput. Surv.}, vol.~23, no.~1, p. 5–48, Mar. 1991. [Online]. Available: \url{https://doi.org/10.1145/103162.103163}
\BIBentrySTDinterwordspacing

\bibitem{nvidia_nvidia_2025}
\BIBentryALTinterwordspacing
NVIDIA, ``\BIBforeignlanguage{en}{{NVIDIA} {Blackwell} {Architecture} {Technical} {Overview}},'' NVIDIA, Tech. Rep., 2025. [Online]. Available: \url{https://resources.nvidia.com/en-us-blackwell-architecture}
\BIBentrySTDinterwordspacing

\bibitem{advanced_micro_devices_inc_amd_2025}
\BIBentryALTinterwordspacing
I.~Advanced Micro~Devices, ``{AMD} {CDNA} 4 {Architecture},'' AMD, Tech. Rep., Oct. 2025. [Online]. Available: \url{https://www.amd.com/content/dam/amd/en/documents/instinct-tech-docs/white-papers/amd-cdna-4-architecture-whitepaper.pdf}
\BIBentrySTDinterwordspacing

\bibitem{micikevicius2022fp8formatsdeeplearning}
\BIBentryALTinterwordspacing
P.~Micikevicius, D.~Stosic, N.~Burgess, M.~Cornea, P.~Dubey, R.~Grisenthwaite, S.~Ha, A.~Heinecke, P.~Judd, J.~Kamalu, N.~Mellempudi, S.~Oberman, M.~Shoeybi, M.~Siu, and H.~Wu, ``Fp8 formats for deep learning,'' 2022. [Online]. Available: \url{https://arxiv.org/abs/2209.05433}
\BIBentrySTDinterwordspacing

\bibitem{kuzmin_fp8_2022}
\BIBentryALTinterwordspacing
A.~Kuzmin, M.~van Baalen, Y.~Ren, M.~Nagel, J.~Peters, and T.~Blankevoort, ``\BIBforeignlanguage{en}{{FP8} {Quantization}: {The} {Power} of the {Exponent}},'' \emph{\BIBforeignlanguage{en}{Advances in Neural Information Processing Systems}}, vol.~35, pp. 14\,651--14\,662, Dec. 2022. [Online]. Available: \url{https://proceedings.neurips.cc/paper_files/paper/2022/hash/5e07476b6bd2497e1fbd11b8f0b2de3c-Abstract-Conference.html}
\BIBentrySTDinterwordspacing

\bibitem{rouhani_microscaling_2023}
\BIBentryALTinterwordspacing
B.~D. Rouhani, R.~Zhao, A.~More, M.~Hall, A.~Khodamoradi, S.~Deng, D.~Choudhary, M.~Cornea, E.~Dellinger, K.~Denolf, S.~Dusan, V.~Elango, M.~Golub, A.~Heinecke, P.~James-Roxby, D.~Jani, G.~Kolhe, M.~Langhammer, A.~Li, L.~Melnick, M.~Mesmakhosroshahi, A.~Rodriguez, M.~Schulte, R.~Shafipour, L.~Shao, M.~Siu, P.~Dubey, P.~Micikevicius, M.~Naumov, C.~Verrilli, R.~Wittig, D.~Burger, and E.~Chung, ``Microscaling {Data} {Formats} for {Deep} {Learning},'' Oct. 2023, arXiv:2310.10537 [cs]. [Online]. Available: \url{http://arxiv.org/abs/2310.10537}
\BIBentrySTDinterwordspacing

\bibitem{googleBFloat16Secret}
``{B}{F}loat16: {T}he secret to high performance on {C}loud {T}{P}{U}s | {G}oogle {C}loud {B}log --- cloud.google.com,'' \url{https://cloud.google.com/blog/products/ai-machine-learning/bfloat16-the-secret-to-high-performance-on-cloud-tpus}, [Accessed 01-12-2025].

\bibitem{bordawekar2022efloatentropycodedfloatingpoint}
\BIBentryALTinterwordspacing
R.~Bordawekar, B.~Abali, and M.-H. Chen, ``Efloat: Entropy-coded floating point format for compressing vector embedding models,'' 2022. [Online]. Available: \url{https://arxiv.org/abs/2102.02705}
\BIBentrySTDinterwordspacing

\bibitem{hao_neuzip_2024}
\BIBentryALTinterwordspacing
Y.~Hao, Y.~Cao, and L.~Mou, ``{NeuZip}: {Memory}-{Efficient} {Training} and {Inference} with {Dynamic} {Compression} of {Neural} {Networks},'' Oct. 2024, arXiv:2410.20650 [cs]. [Online]. Available: \url{http://arxiv.org/abs/2410.20650}
\BIBentrySTDinterwordspacing

\bibitem{shannon_communication_1949}
\BIBentryALTinterwordspacing
C.~Shannon, ``Communication in the {Presence} of {Noise},'' \emph{Proceedings of the IRE}, vol.~37, no.~1, pp. 10--21, Jan. 1949. [Online]. Available: \url{https://ieeexplore.ieee.org/abstract/document/1697831}
\BIBentrySTDinterwordspacing

\bibitem{WolpertStochThermo2024}
\BIBentryALTinterwordspacing
D.~H. Wolpert, J.~Korbel, C.~W. Lynn, F.~Tasnim, J.~A. Grochow, G.~Kardeş, J.~B. Aimone, V.~Balasubramanian, E.~D. Giuli, D.~Doty, N.~Freitas, M.~Marsili, T.~E. Ouldridge, A.~W. Richa, P.~Riechers, Édgar Roldán, B.~Rubenstein, Z.~Toroczkai, and J.~Paradiso, ``Is stochastic thermodynamics the key to understanding the energy costs of computation?'' \emph{Proceedings of the National Academy of Sciences}, vol. 121, no.~45, p. e2321112121, 2024. [Online]. Available: \url{https://www.pnas.org/doi/abs/10.1073/pnas.2321112121}
\BIBentrySTDinterwordspacing

\bibitem{Wolpert2019}
D.~H. Wolpert, ``The stochastic thermodynamics of computation,'' \emph{Journal of Physics A: Mathematical and Theoretical}, vol.~52, no.~19, p. 193001, 2019.

\bibitem{Landauer1961}
R.~Landauer, ``Irreversibility and heat generation in the computing process,'' \emph{IBM Journal of Research and Development}, vol.~5, no.~3, pp. 183--191, Jul 1961.

\bibitem{maroney2008physicalbasisgibbsvonneumann}
\BIBentryALTinterwordspacing
O.~J.~E. Maroney, ``The physical basis of the gibbs-von neumann entropy,'' 2008. [Online]. Available: \url{https://arxiv.org/abs/quant-ph/0701127}
\BIBentrySTDinterwordspacing

\bibitem{MaroneyGeneralizingLandauer}
\BIBentryALTinterwordspacing
------, ``Generalizing landauer's principle,'' \emph{Phys. Rev. E}, vol.~79, p. 031105, Mar 2009. [Online]. Available: \url{https://link.aps.org/doi/10.1103/PhysRevE.79.031105}
\BIBentrySTDinterwordspacing

\bibitem{dambrosiaThermodynamicCostsSimple2026}
\BIBentryALTinterwordspacing
S.~H. D'Ambrosia, S.~M. Daniels, M.~R. DeWeese, and A.~Sahai. The {{Thermodynamic Costs}} of {{Simple Linear Regression}}. [Online]. Available: \url{http://arxiv.org/abs/2605.19195}
\BIBentrySTDinterwordspacing

\bibitem{nairFeedbackControlData2007}
\BIBentryALTinterwordspacing
G.~N. Nair, F.~Fagnani, S.~Zampieri, and R.~J. Evans, ``Feedback {{Control Under Data Rate Constraints}}: {{An Overview}},'' vol.~95, no.~1, pp. 108--137. [Online]. Available: \url{https://ieeexplore.ieee.org/abstract/document/4118468}
\BIBentrySTDinterwordspacing

\bibitem{kostinaRatelimitedControlSystems2016}
\BIBentryALTinterwordspacing
V.~Kostina, Y.~Peres, M.~Z. Rácz, and G.~Ranade, ``Rate-limited control of systems with uncertain gain,'' in \emph{2016 54th {{Annual Allerton Conference}} on {{Communication}}, {{Control}}, and {{Computing}} ({{Allerton}})}, pp. 1189--1196. [Online]. Available: \url{https://ieeexplore.ieee.org/abstract/document/7852370}
\BIBentrySTDinterwordspacing

\bibitem{kostinaRateCostTradeoffsControl2019}
\BIBentryALTinterwordspacing
V.~Kostina and B.~Hassibi, ``Rate-{{Cost Tradeoffs}} in {{Control}},'' vol.~64, no.~11, pp. 4525--4540. [Online]. Available: \url{https://ieeexplore.ieee.org/document/8693967}
\BIBentrySTDinterwordspacing

\bibitem{hanQuantizedOnlineLQR2026}
\BIBentryALTinterwordspacing
B.~Han, V.~Kostina, and B.~Hassibi. Quantized {{Online LQR}}. [Online]. Available: \url{http://arxiv.org/abs/2604.11930}
\BIBentrySTDinterwordspacing

\bibitem{eliaStabilizationLinearSystems2001}
\BIBentryALTinterwordspacing
N.~Elia and S.~Mitter, ``Stabilization of linear systems with limited information,'' vol.~46, no.~9, pp. 1384--1400. [Online]. Available: \url{https://ieeexplore.ieee.org/abstract/document/948466}
\BIBentrySTDinterwordspacing

\bibitem{elia2002QuantizedStabilization}
N.~Elia and E.~Frazzoli, ``Quantized stabilization of two-input linear systems: A lower bound on the minimal quantization density,'' in \emph{Hybrid Systems: {{Computation}} and Control}, C.~J. Tomlin and M.~R. Greenstreet, Eds.\hskip 1em plus 0.5em minus 0.4em\relax Springer Berlin Heidelberg, pp. 179--193.

\bibitem{okanoMinimumDataRate2014}
\BIBentryALTinterwordspacing
K.~Okano and H.~Ishii. Minimum data rate for stabilization of linear systems with parametric uncertainties. [Online]. Available: \url{http://arxiv.org/abs/1405.5932}
\BIBentrySTDinterwordspacing

\bibitem{GrayNeuhoff1998Quantization}
R.~M. Gray and D.~L. Neuhoff, ``Quantization,'' \emph{IEEE Transactions on Information Theory}, vol.~44, no.~6, pp. 2325--2383, 1998.

\bibitem{gish_asymptotically_1968}
\BIBentryALTinterwordspacing
H.~Gish and J.~Pierce, ``Asymptotically efficient quantizing,'' \emph{IEEE Transactions on Information Theory}, vol.~14, no.~5, pp. 676--683, Sep. 1968. [Online]. Available: \url{https://ieeexplore.ieee.org/abstract/document/1054193}
\BIBentrySTDinterwordspacing

\bibitem{kostina_data_2017}
\BIBentryALTinterwordspacing
V.~Kostina, ``Data {Compression} {With} {Low} {Distortion} and {Finite} {Blocklength},'' \emph{IEEE Transactions on Information Theory}, vol.~63, no.~7, pp. 4268--4285, Jul. 2017. [Online]. Available: \url{https://ieeexplore.ieee.org/abstract/document/7867787}
\BIBentrySTDinterwordspacing

\bibitem{CoverThomas}
T.~M. Cover and J.~A. Thomas, \emph{Elements of Information Theory (Wiley Series in Telecommunications and Signal Processing)}.\hskip 1em plus 0.5em minus 0.4em\relax USA: Wiley-Interscience, 2006.

\end{thebibliography}

\newpage
\appendices
\onecolumn

\section{Bounds on the KL divergence term}\label{app:kl_bounds}



Now to prove Theorem~\ref{thm:continuum_bounds_main}. Its statement is copied here for the reader's convenience.

\begingroup
\addtocounter{theorem}{-1}
\renewcommand{\thetheorem}{\ref{thm:continuum_bounds_main}}
\begin{theorem}[Bounds on the KL divergence] Under the convention $0\log 0 = 0$, assume
\begin{align}
    \Lambda_\ib
    \triangleq
    \operatorname*{ess\,sup}_{\xv\in \Bi} \Ri(\xv)
    =
    \frac{1}{\gi}
    \operatorname*{ess\,sup}_{\xv\in \Bi} f_\Xv(\xv) < \infty.
\end{align}

Then,
\begin{align}
    &\sum\limits_{\ib}p_\ib
\max_{\lambda \geq 1}\bracket{
\lambda\log \lambda
-
(\lambda-1)\log e
}
\Li(\lambda)
\leq \relent{f_\Xv}{g}
    \le
    \sum\limits_{\ib}p_\ib
    \int_1^{\Lambda_\ib}
    \paren{\log \lambda+\log e}
    \Li(\lambda)
    \,d\lambda.
\end{align}
\end{theorem}
\endgroup

\begin{proof}
The KL divergence can be written as
\begin{align}
    &\relent{f_\Xv}{g} 
    = \int
    f_\Xv(\xv)
    \log\frac{f_\Xv(\xv)}{g(\xv)}
    \,d\xv
    = \sum\limits_{\ib}\int_{\Bi}
    f_\Xv(\xv)
    \log\frac{f_\Xv(\xv)}{\gi}
    \,d\xv.
\end{align}
To prove the upper bound, since for all $\ib$, when $\xv \in \Bi$,
\begin{align}
    f_\Xv(\xv)
    \log\frac{f_\Xv(\xv)}{\gi}
    \leq 0
    \quad
    \text{whenever }
    f_\Xv(\xv)\leq \gi,
\end{align}
we have
\begin{align}
    &\int_{\Bi}
    f_\Xv(\xv)
    \log\frac{f_\Xv(\xv)}{\gi}
    \,d\xv
    \le
    \int_{\curly{\Ri\ge 1}}
    f_\Xv(\xv)
    \log \Ri(\xv)
    \,d\xv
    =
    \gi
    \int_{\Bi}
    \Ri(\xv)\log \Ri(\xv)
    \indic{\Ri(\xv)\ge 1}
    \,d\xv.
\end{align}
For $r\ge 1$,
\begin{align}
    r\log r
    =
    \int_1^r
    \log \lambda
    \,d\lambda + (r-1)\log e.
\end{align}
Thus, for every $\xv\in \Bi$,
\begin{align}
    &\Ri(\xv)\log \Ri(\xv)
    \indic{\Ri(\xv)\ge 1}
    =\int_1^{\Ri(\xv)}
    \paren{\log \lambda+\log e}
    \,d\lambda\indic{\Ri(\xv)\ge 1}
    =
    \int_1^{\Lambda_\ib}\indic{\lambda \leq \Ri(\xv)}
    \paren{\log \lambda+\log e}
    \,d\lambda.
\end{align}
Since for $\lambda \geq 1$, $\log\lambda \geq 0$, Tonelli's theorem can be applied to obtain the bound,
\begin{align}
    &\int_{\Bi}
    f_\Xv(\xv)
    \log\frac{f_\Xv(\xv)}{\gi}
    \,d\xv
    \le
    \gi
    \int_{\Bi}
    \int_1^{\Lambda_\ib}
    \paren{\log \lambda+\log e}
    \indic{\Ri(\xv)\ge \lambda}
    \,d\lambda
    \,d\xv
    \\
    &
    =
    \gi
    \int_1^{\Lambda_\ib}
    \paren{\log \lambda+\log e}
    \int_{\Bi}
    \indic{\Ri(\xv)\ge \lambda}
    \,d\xv
    \,d\lambda
    =
    \gi
    \int_1^{\Lambda_\ib}
    \paren{\log \lambda+\log e}
    \abs{
        \curly{
            \xv\in \Bi:
            \Ri(\xv)\ge \lambda
        }
    }
    \,d\lambda
    \\
    &
    =
    \gi \abs{\Bi}
    \int_1^{\Lambda_\ib}
    \paren{\log \lambda+\log e}
    \Li(\lambda)
    \,d\lambda
    =
    p_\ib
    \int_1^{\Lambda_\ib}
    \paren{\log \lambda+\log e}
    \Li(\lambda)
    \,d\lambda.
\end{align}
Summing over $\ib$ yields the upper bound on $\relent{f_\Xv}{g}$.

Now to prove the lower bound, let $\phi(r) \triangleq r\log r - (r - 1)\log e$. Taking its first and second derivatives,
\begin{align}
\phi'(r)
&=
\log r, \qquad \phi''(r) =
\frac{\log e}{r},
\end{align}
so $\phi$ is decreasing on $\paren{0,1}$, increasing on
$\paren{1,\infty}$, and $
\phi(1) = 0
$.
Thus
$
\phi(r) \ge 0$
for all $r\ge 0$.
Moreover, for
$\lambda\ge 1$, monotonicity of $\phi$ on $[1,\infty)$ gives
$
\phi(r)
\ge
\phi(\lambda)$ whenever $r\ge \lambda$.
Therefore,
\begin{align}
\phi\paren{\Ri(\xv)}
&\ge
\phi(\lambda)
\indic{\Ri(\xv)\ge \lambda}.
\end{align}
By the monotonicity of expectation, we have
\begin{align}
&p_\ib
\Esub{
\phi\paren{\Ri(\Xv)}
}{\mathrm{Unif}(\Bi)}
\ge
p_\ib
\phi(\lambda)
\Esub{
\indic{\Ri(\Xv)\ge \lambda}
}{\mathrm{Unif}(\Bi)}
=
p_\ib
\phi(\lambda)
\frac{1}{\abs{\Bi}}
\abs{\curly{\xv\in\Bi:\Ri(\xv)\ge \lambda}}                 
=
p_\ib
\phi(\lambda)
\Li(\lambda).\label{eq:mon_expec}
\end{align}
We know that the per-bin KL divergence can be written as \begin{align}
    \int_{\Bi}
    f_\Xv(\xv)
    \log\frac{f_\Xv(\xv)}{\gi}
    \,d\xv =
\gi \int_{\Bi} \Ri(\xv)\log\Ri(\xv)\,d\xv = p_\ib
\Esub{\Ri(\Xv)\log\Ri(\Xv)
}{\mathrm{Unif}(\Bi)}.
\end{align}
Finally, substituting the definition of $\phi$ into Eq.~\eqref{eq:mon_expec} and using Proposition~\ref{lem:feasibility} (which gives $\Esub{\Ri(\Xv)}{\mathrm{Unif}(\Bi)} = 1$) yields
\begin{align}
&p_\ib
\Esub{\Ri(\Xv)\log\Ri(\Xv) - \paren{\Ri(\Xv) - 1}\log e 
}{\mathrm{Unif}(\Bi)}
= p_\ib
\Esub{\Ri(\Xv)\log\Ri(\Xv)
}{\mathrm{Unif}(\Bi)}
\ge
p_\ib
\bracket{
\lambda\log\lambda
-
(\lambda-1)\log e
}
\Li(\lambda).
\end{align}
This holds for every $\lambda \ge 1$, so \begin{align}
    \relent{f_\Xv}{g} \geq \sum_{\ib}p_\ib \max_{\lambda \geq 1}
\bracket{
\lambda\log\lambda
-
(\lambda-1)\log e
}
\Li(\lambda).
\end{align}
\end{proof}
\section{The Entropy for Commonly Distributed Random Variables}\label{app:approx_dists}
We will proceed to closed-form formulas for the discrete entropy of commonly distributed random variables represented in floating-point, applying Eq.~\eqref{eq:approx_discussion}. Throughout, let $p$ be the precision of the floating-point number, the differential entropy be
\begin{equation}
    h(X) \triangleq - \int f_X(x) \log [f_X(x)] dx,
\end{equation}
and the approximate bin size be
\begin{equation}
    \Delta_s(x) \triangleq \frac{1}{\sqrt{2}} |x| \cdot 2^{1-p}.
\end{equation}

\subsection{Centered Univariate Gaussian Distribution}
Let $X \sim \mathcal{N}(0, \sigma_x^2)$. The approximate entropy of the floating-point quantized discrete random variable $X_{fp}$ is given by
    \begin{equation}
    \begin{aligned}
    \label{eq:1d-0m-fp}
        \Ht_{s}(\Xfp) &\triangleq h(X) - \Esub{\log\paren{\Delta_s(X)}}{f_X} \\
        & = \frac{1}{2}\log[2\pi e \sigma^2_x] - \int^{\infty}_{-\infty} \frac{\exp(-\frac{x^2}{2\sigma_x^2})}{\sigma_x \sqrt{2\pi}} \log [\Delta_s(x)] \, dx  \\
        & = \frac{1}{2}\log[2\pi e \sigma^2_x] - \int^{\infty}_{-\infty} \frac{\exp(-\frac{x^2}{2\sigma_x^2})}{\sigma_x \sqrt{2\pi}} (\log[|x|/\sqrt{2}] - (p - 1)) \, dx \\
        & = \frac{1}{2}\log[2\pi e \sigma^2_x] + (p - 1) - 2\int^{\infty}_0 \frac{\exp(-\frac{x^2}{2\sigma_x^2})}{\sigma_x \sqrt{2\pi}} \log[x/\sqrt{2}] \, dx\\
        & = \frac{1}{2}\log[2\pi e \sigma^2_x] + (p-1) + \bigg(1 + \frac{\gamma_e}{2\ln[2]} -\frac{1}{2} \log[\sigma^2_x] \bigg)\\
        & = (p-1) + \frac{1}{2} \log[2 \pi e] + 1 + \frac{\gamma_e}{2\ln[2]}\\
        & = p + \frac{1}{2} \log[2 \pi e] + \frac{\gamma_e}{2\ln[2]},
    \end{aligned}
\end{equation}
where $\gamma_e$ is Euler's constant. 

\subsection{Uniform Distribution}
Let $X \sim \unif{a}{b}$ with $a, b \in \R{}$ and $a < b$. The approximate floating-point entropy is
\begin{equation}
    \begin{aligned}
        \tilde H_{s}(X_{fp}) &= h(X) - \Esub{\log\paren{\Delta_s(X)}}{f_X}\\
        &= \log(b - a) - \frac{1}{b-a}\int\limits_a^b \log\paren{\frac{|x|}{\sqrt{2}}} + (1-p)dx\\
        & = \log(b - a) +p - 1 + \log(e) \\
        &\quad +\log\sqrt{2} + \frac{a\log|a| - b\log |b|}{b-a} \\
        & = p - 1  +\log(\sqrt{2}e(b-a)) + \frac{a\log|a| - b\log |b|}{b-a}.
    \end{aligned}
\end{equation}
This implies that for a centered uniform distribution, where $a = -b$ and $b > 0$,
\begin{equation}
    \begin{aligned}
        \tilde H_{s}(X_{fp}) & = p - 1  +\log(2\sqrt{2}e b) + \frac{-b\log b - b\log b}{2b} \\
        & = p - 1  +\log(2\sqrt{2}e b) - \log b \\
        & = p  + \frac{1}{2} + \log(e).
    \end{aligned}
\end{equation}

\subsection{Gamma Distribution}
Let $X \sim \text{Gamma}(\alpha,\theta)$, where $\alpha$ is the shape parameter, and $\theta$ is the scale parameter. In this case, the approximate floating-point entropy is even simpler to evaluate since $X > 0 $, meaning the expectation over the bin size is just given by the first log moment. With $\psi(\cdot)$ as the digamma function, and $\Gamma(\cdot)$ being the gamma function, we have
\begin{equation}
    h(X) = \log(e)(\alpha + \psi(\alpha)(1-\alpha)) + \log(\theta) + \log(\Gamma (\alpha)),
\end{equation}
and
\begin{equation}
    -\Esub{\log\paren{\Delta_s(X)}}{f_X} = p-1 - \psi(\alpha) \log(e) - \log(\theta) + \log(\sqrt{2}).
\end{equation}
From this we see
\begin{equation}
    \begin{aligned}
        \tilde H_{s}(X_{fp}) & = p - \frac{1}{2} + \alpha \log(e)(1 - \psi(\alpha) ) + \log(\Gamma(\alpha)).
    \end{aligned}
\end{equation}

\subsection{Chi-squared Distribution}
Let $X \sim \chi^2(k)$. The approximate floating-point entropy can be determined by the expression for the Gamma distribution, as $\chi^2(k)$ is the Gamma distribution with $\alpha = k/2$, $\theta = 2$.
\begin{equation}
    \begin{aligned}
        \tilde H_{s}(X_{fp}) &= p - \frac{1}{2} + \frac{k}{2} \log(e) \bigg (1-\psi \bigg(\frac{k}{2} \bigg) \bigg) + \log \bigg (\Gamma \bigg (\frac{k}{2} \bigg ) \bigg),
    \end{aligned}
\end{equation}
where again $\psi(\cdot)$ is the digamma function, and $\Gamma(\cdot)$ is the gamma function.

\subsection{Centered Laplace Distribution}
Let $X \sim$ Laplace$(0,b)$ with scale parameter $b>0$. The approximate floating-point entropy is
\begin{equation}
    \begin{aligned}
        \tilde H_{s}(X_{fp}) &= h(X) - \Esub{\log\paren{\Delta_s(X)}}{f_X}\\
        &= h(X) - \Esub{\log\abs{X}}{f_X} + p - \frac{1}{2}.
    \end{aligned}
\end{equation}
The differential entropy for the Laplace distribution is
\begin{equation}
    h(X) = \log(2be),
\end{equation}
while the first log moment of $|X|$ is 
\begin{equation}
    \Esub{\log\abs{X}}{f_X} = \log(b) - \frac{\gamma_e}{\ln 2},
\end{equation}
meaning
\begin{equation}
    \begin{aligned}
        \tilde H_{s}(X_{fp}) &=p+\frac{1}{2}+ \frac{1+\gamma_e}{\ln2}.
    \end{aligned}
\end{equation}

\subsection{Centered Logistic Distribution}
Let $X \sim$ Logistic$(0,s)$ with scale parameter $s>0$. The approximate floating-point entropy is
\begin{equation}
    \begin{aligned}
        \tilde H_{s}(X_{fp}) &= h(X) - \Esub{\log\paren{\Delta_s(X)}}{f_X}\\
        &= h(X) - \Esub{\log\abs{X}}{f_X} + p - \frac{1}{2}.
    \end{aligned}
\end{equation}
The differential entropy for the Logistic distribution is
\begin{equation}
    h(X) = \log(s) + \frac{2}{\ln2},
\end{equation}
while the first log moment of $|X|$ is 
\begin{equation}
    \Esub{\log\abs{X}}{f_X} = \log(s) + \frac{\ln(\pi/2)- \gamma_e}{\ln2},
\end{equation}
meaning
\begin{equation}
    \begin{aligned}
        \tilde H_{s}(X_{fp}) &=p-\frac{1}{2}+ \frac{2+\gamma_e - \ln(\pi/2)}{\ln2}.
    \end{aligned}
\end{equation}

\subsection{Weibull Distribution}
Let $X \sim$ Weibull($\lambda,k$). The approximate floating-point entropy is
\begin{equation}
    \begin{aligned}
        \tilde H_{s}(X_{fp}) &= h(X) - \Esub{\log\paren{\Delta_s(X)}}{f_X}\\
        &= h(X) - \Esub{\log\abs{X}}{f_X} + p - \frac{1}{2}.
    \end{aligned}
\end{equation}
The differential entropy for the Weibull distribution is
\begin{equation}
    h(X) = \log \bigg( \frac{\lambda}{k} \bigg) + \frac{1 + \gamma_e(1-\frac{1}{k})}{\ln 2},
\end{equation}
while the first log moment of $|X|$ is 
\begin{equation}
    \Esub{\log\abs{X}}{f_X} = \log(\lambda) - \frac{\gamma_e}{k\ln2},
\end{equation}
meaning
\begin{equation}
    \begin{aligned}
        \tilde H_{s}(X_{fp}) &=p-\frac{1}{2}+ \frac{1+\gamma_e}{\ln2} - \log(k).
    \end{aligned}
\end{equation}

\subsection{Lognormal Distribution}
Let $X \sim$ Lognormal($\mu, \sigma^2$). Note that for the Lognormal distribution, $\mu$ is the scale parameter and $\sigma$ is a shape parameter. The approximate floating-point entropy is
\begin{equation}
    \begin{aligned}
        \tilde H_{s}(X_{fp}) &= h(X) - \Esub{\log\paren{\Delta_s(X)}}{f_X}\\
        &= h(X) - \Esub{\log\abs{X}}{f_X} + p - \frac{1}{2}.
    \end{aligned}
\end{equation}
The differential entropy for the Lognormal distribution is
\begin{equation}
    h(X) = \frac{\mu}{\ln2} + \log (\sigma \sqrt{2 \pi e}),
\end{equation}
while the first log moment of $|X|$ is 
\begin{equation}
    \Esub{\log\abs{X}}{f_X} = \frac{\mu}{\ln2},
\end{equation}
meaning
\begin{equation}
    \begin{aligned}
        \tilde H_{s}(X_{fp}) &=p-\frac{1}{2}+ \log(\sigma \sqrt{2 \pi e}).
    \end{aligned}
\end{equation}

\subsection{Pareto Distribution}
Let $X \sim$ Pareto($x_m,\alpha$). The approximate floating-point entropy is
\begin{equation}
    \begin{aligned}
        \tilde H_{s}(X_{fp}) &= h(X) - \Esub{\log\paren{\Delta_s(X)}}{f_X}\\
        &= h(X) - \Esub{\log\abs{X}}{f_X} + p - \frac{1}{2}.
    \end{aligned}
\end{equation}
The differential entropy for the Pareto distribution is
\begin{equation}
    h(X) = \log \bigg( \frac{x_m}{\alpha}\bigg) + \frac{1+\frac{1}{\alpha}}{\ln2},
\end{equation}
while the first log moment of $|X|$ is 
\begin{equation}
    \Esub{\log\abs{X}}{f_X} = \log(x_m) + \frac{1}{\alpha \ln2},
\end{equation}
meaning
\begin{equation}
    \begin{aligned}
        \tilde H_{s}(X_{fp}) &=p-\frac{1}{2}+ \log \bigg( \frac{e}{\alpha}\bigg ).
    \end{aligned}
\end{equation}

\subsection{Beta Distribution}

Let $X \sim \mathrm{Beta}(\alpha,\beta)$, with $\alpha,\beta>0$. Since the beta distribution has support on $(0,1)$, we have $|X|=X$. The approximate floating-point entropy is
\begin{equation}
    \begin{aligned}
        \tilde H_{s}(X_{fp}) 
        &= h(X) - \Esub{\log\paren{\Delta_s(X)}}{f_X}\\
        &= h(X) - \Esub{\log X}{f_X} + p - \frac{1}{2}.
    \end{aligned}
\end{equation}
The differential entropy for the beta distribution is
\begin{equation}
    h(X)
    =
    \frac{
    \ln B(\alpha,\beta)
    -(\alpha-1)\psi(\alpha)
    -(\beta-1)\psi(\beta)
    +(\alpha+\beta-2)\psi(\alpha+\beta)
    }{\ln 2},
\end{equation}
where $B(\alpha,\beta)$ is the beta function and $\psi$ is the digamma function. The first log moment is
\begin{equation}
    \Esub{\log X}{f_X}
    =
    \frac{\psi(\alpha)-\psi(\alpha+\beta)}{\ln 2}.
\end{equation}
Therefore,
\begin{equation}
    \begin{aligned}
        \tilde H_{s}(X_{fp})
        &=
        p-\frac{1}{2}
        +
        \frac{
        \ln B(\alpha,\beta)
        -(\alpha-1)\psi(\alpha)
        -(\beta-1)\psi(\beta)
        +(\alpha+\beta-2)\psi(\alpha+\beta)
        -\psi(\alpha)
        +\psi(\alpha+\beta)
        }{\ln 2} \\
        &=
        p-\frac{1}{2}
        +
        \log B(\alpha,\beta)+
        \frac{(\alpha+\beta-1)\psi(\alpha+\beta)
        -\alpha\psi(\alpha)
        -(\beta-1)\psi(\beta)
        }{\ln 2}.
    \end{aligned}
\end{equation}

\subsection{Student's $t$ Distribution}

Let $X \sim t_{\nu}(0,s)$ be a centered Student's $t$ distribution with degrees of freedom $\nu>0$ and scale parameter $s>0$. The approximate floating-point entropy is
\begin{equation}
    \begin{aligned}
        \tilde H_{s}(X_{fp}) 
        &= h(X) - \Esub{\log\paren{\Delta_s(X)}}{f_X}\\
        &= h(X) - \Esub{\log\abs{X}}{f_X} + p - \frac{1}{2}.
    \end{aligned}
\end{equation}
The differential entropy for the Student's $t$ distribution is
\begin{equation}
    h(X)
    =
    \frac{1}{\ln 2}
    \left[
    \ln\left(
    s\sqrt{\nu}\,
    B\left(\frac{\nu}{2},\frac{1}{2}\right)
    \right)
    +
    \frac{\nu+1}{2}
    \left(
    \psi\left(\frac{\nu+1}{2}\right)
    -
    \psi\left(\frac{\nu}{2}\right)
    \right)
    \right],
\end{equation}
where $B$ is the beta function and $\psi$ is the digamma function. The first absolute log moment is
\begin{equation}
    \Esub{\log\abs{X}}{f_X}
    =
    \frac{1}{\ln 2}
    \left[
    \ln s
    +
    \frac{1}{2}
    \left(
    \ln \nu
    +
    \psi\left(\frac{1}{2}\right)
    -
    \psi\left(\frac{\nu}{2}\right)
    \right)
    \right].
\end{equation}
Therefore,
\begin{equation}
    \begin{aligned}
        \tilde H_{s}(X_{fp})
        &=
        p-\frac{1}{2}
        +
        \frac{1}{\ln 2}
        \bigg[
        \ln\left(
        s\sqrt{\nu}\,
        B\left(\frac{\nu}{2},\frac{1}{2}\right)
        \right)
        +
        \frac{\nu+1}{2}
        \left(
        \psi\left(\frac{\nu+1}{2}\right)
        -
        \psi\left(\frac{\nu}{2}\right)
        \right) \\
        &\qquad\qquad
        -
        \ln s
        -
        \frac{1}{2}
        \left(
        \ln \nu
        +
        \psi\left(\frac{1}{2}\right)
        -
        \psi\left(\frac{\nu}{2}\right)
        \right)
        \bigg] \\
        &=
        p-\frac{1}{2}
        +\log B\left(\frac{\nu}{2},\frac{1}{2}\right)+
        \frac{\frac{\nu+1}{2}\psi\left(\frac{\nu+1}{2}\right)
        -
        \frac{\nu}{2}\psi\left(\frac{\nu}{2}\right)
        -
        \frac{1}{2}\psi\left(\frac{1}{2}\right)
        }{\ln 2}.
    \end{aligned}
\end{equation}

\subsection{Centered Multivariate Gaussian Distribution}
Let $(X_1, \dots, X_d) = \Xv \sim \gauss{0}{\Sigma}$ where $\Sigma \succ 0$. Let $|\Sigma|$ denote the determinant of $\Sigma$ and let $\lambda_i(\Sigma)$ denote the $i$-th eigenvalue of $\Sigma$.
Assuming each variable is independently quantized on its own floating-point number, the approximate joint floating-point entropy is
\begin{equation}
    \begin{aligned}
        \tilde H_{s}(\Xv_{fp}) &= h(\Xv) - \Esub{\log\paren{\prod\limits_{i=1}^{d}\Delta_s(X_i)}}{f_\Xv}\\
        &=\frac{1}{2}\log\paren{(2\pi e)^d|\Sigma|} - \sum\limits_{i=1}^{d}\Esub{\log\paren{\Delta_s(X_i)}}{f_\Xv}\\
        &=\frac{1}{2}\log\paren{(2\pi e)^d|\Sigma|} - \sum\limits_{i=1}^{d}\Esub{\log\paren{\Delta_s(X_i)}}{f_{X_i}}\\
        &= \frac{1}{2}\log\paren{(2\pi e)^d|\Sigma|} \\
        &\quad - \sum\limits_{i=1}^{d}\bracket{-(p-1) - \paren{1 + \frac{\gamma_e}{2\ln[2]} -\frac{1}{2} \log\paren{\Sigma_{ii}}}}\\
        &= \frac{1}{2}\log\paren{(2\pi e)^d|\Sigma|} + d\paren{p + \frac{\gamma_e}{2\ln[2]}} -\frac{1}{2} \sum\limits_{i=1}^{d}\log\paren{\Sigma_{ii}}\\
        &= d\paren{\frac{1}{2}\log\paren{2\pi e} + \paren{p + \frac{\gamma_e}{2\ln[2]}}} + \frac{1}{2} \log\paren{\frac{\prod\limits_{i=1}^{d}\lambda_i(\Sigma)}{\prod\limits_{i=1}^{d}\Sigma_{ii}}}.
    \end{aligned}
\end{equation}
Notice that the last term is zero when $\Sigma$ is diagonal.
\section{Bound on error from outer clipping bins and exponent boundary bins}
\label{app:bin_bounds}
Here we provide the full proof of Theorem~\ref{thm:steps}. Its statement is copied here for the reader's convenience.
\begingroup
\addtocounter{theorem}{-1}
\renewcommand{\thetheorem}{\ref{thm:steps}}
\begin{theorem}[Error bound for smoothing and extending the bin-size function]
Let $\Xv\sim f_\Xv$ be a $d$-dimensional random vector with the probability density $f_\Xv$. Let $\mathbb{O}^d \triangleq \mathbb{R}^d \setminus \mathbb{U}^d$ be the overflow region, and the underflow region be $\mathbb{S}^d \triangleq [-2^{e_{min}}, 2^{e_{min}}]^d$. Then
\begin{equation}
\abs{\Ht(\Xv_{fp}) - \Ht_s(\Xv_{fp})} \leq \frac{d}{2} + \epsilon,
\end{equation}
where 
\begin{equation}
    \epsilon = \abs{
    \int_{\mathbb{S}^d \cup \mathbb{O}^d}
    f_{\Xv}(\xv)
    \log\left[
    f_{\Xv}(\xv)\bDeltas(\xv)
    \right]
    d\xv}.
\end{equation}
\end{theorem}
\endgroup

\begin{proof}
There are three types of bins for which we will bound the error: bins on the interior of an exponent block, bins on the outer edges, and bins on the boundary between exponent blocks (regions for which a single exponent value applies). The bound can be shown by considering the ratio between bin sizes, case by case.

\paragraph{Interior bins}
For interior bins strictly within each exponent block $e$, by Eq.~\eqref{eq:true_bins},
\begin{equation}
    \label{eq:real_bins_appendix}
    \Delta(x_j) = 2^{e-(p-1)} \quad \text{for } 2^e \leq |x_j| < 2^{e+1},
\end{equation}
and within this range $\Delta_s(x_j) = \frac{|x_j|}{\sqrt{2}} \cdot 2^{1-p}$ satisfies
\begin{equation}
    \begin{aligned}
    &\frac{2^e}{\sqrt{2}}2^{1-p} \leq \Delta_s(x_j) \leq \frac{2^{e+1}}{\sqrt{2}}2^{1-p}\\
    \implies &\frac{1}{\sqrt{2}} \leq \frac{\Delta_s(x_j)}{\Delta(x_j)} < \sqrt{2}.
    \end{aligned}
\end{equation}

\paragraph{Outer clipping bins}

For the outer clipping bins, by symmetry it suffices to analyze $[\frac{u_K + u_{K-1}}{2}, 2^{e_{\max}+1}-2^{e_{\max}-p}]$ which is the bin on the positive side of the real line. When $p \geq 2$, we know that $u_K = 2^{e_{max}+1}(1 - 2^{-p})$, and $u_{K-1} = 2^{e_{max}+1}(1 - 2^{1-p})$, so the outer clipping bin size is \begin{equation}
\Delta(x_j) = 2^{e_{max} + 1} - 2^{e_{max} - p}- \frac{u_K + u_{K-1}}{2} = 2^{e_{max}}\paren{2^{1-p}} = 2^{e_{max}-p + 1}.
\end{equation}
Now, we have
\begin{equation}
\begin{aligned}
  &\frac{(u_K + u_{K-1})}{2\sqrt{2}}2^{1-p} \leq \Delta_s(x_j) \leq \frac{2^{e_{max} + 1} - 2^{e_{max} - p}}{\sqrt{2}}2^{1-p}\\
  \implies &\frac{ 2^{e_{max}}(1 - 2^{-1-p} - 2^{-p})}{2^{e_{max} - p}\sqrt{2}}2^{1-p} \leq \frac{\Delta_s(x_j)}{\Delta(x_j)} \leq \frac{2^{e_{max}}(1 - 2^{-1 - p})}{2^{e_{max} - p}\sqrt{2}}2^{1-p}\\
  \implies &\frac{ 2(1 - 2^{-1-p} - 2^{-p})}{\sqrt{2}} \leq \frac{\Delta_s(x_j)}{\Delta(x_j)} \leq \frac{2(1- 2^{-1-p})}{\sqrt{2}}\\
  \implies &\sqrt{2}(1 - 2^{-1-p} - 2^{-p}) \leq \frac{\Delta_s(x_j)}{\Delta(x_j)} \leq \sqrt{2}(1- 2^{-1-p}).
\end{aligned}
\end{equation}
When $p = 1$, $u_K = 2^{e_{max}}$ and $u_{K-1} = 2^{e_{max}-1}$. The true bin width is 
\begin{equation}
\Delta(x_j) = 2^{e_{max}+1} - 2^{e_{max}-1} - \frac{1}{2}(2^{e_{max}} + 2^{e_{max}-1}) = 3\cdot2^{e_{max}-2}.
\end{equation}
Using the same bounding technique, we have
\begin{equation}
\begin{aligned}
  &\frac{2^{e_{max}} + 2^{e_{max}-1}}{2\sqrt{2}} \leq \Delta_s(x_j) \leq \frac{2^{e_{max}+1} - 2^{e_{max}-1}}{\sqrt{2}}\\
  \implies &\frac{3\cdot 2^{e_{max}-2}}{3\cdot2^{e_{max}-2}\sqrt{2}} \leq \frac{\Delta_s(x_j)}{\Delta(x_j)} \leq \frac{3\cdot 2^{e_{max}-1}}{3\cdot2^{e_{max}-2}\sqrt{2}}\\
  \implies &\frac{ 1}{\sqrt{2}} \leq \frac{\Delta_s(x_j)}{\Delta(x_j)} \leq \frac{2}{\sqrt{2}}.
\end{aligned}
\end{equation}
For the outer clipping bins, we see that for all $p\geq 1$, $\frac{1}{\sqrt{2}} \leq \frac{\Delta_s(x_j)}{\Delta(x_j)} \leq \sqrt{2}$.

\paragraph{Exponent boundary bins}

At exponent boundaries, let the last representable value of the $e$-th exponent block be $u_A = 2^{e+1} - 2^{e-(p-1)}$, the first representable value of the $e+1$-th exponent block be $u_B = 2^{e+1}$, and the second representable value of the $e+1$-th exponent block be $u_C = 2^{e+1} + 2^{e+1 - (p-1)}$. The left boundary of the midpoint-quantization bin for $u_B$ is $m_L = \frac{u_A + u_B}{2} = 2^{e+1}(1 - 2^{-p-1})$ while the right boundary is $m_R = \frac{u_B + u_C}{2} = 2^{e+1}(1+2^{-p})$. When $x_j \in [m_L, m_R)$, the bin width is
\begin{equation}\label{eq:exp_boundary}
\Delta(x_j) = m_R - m_L = 3 \cdot 2^{e-p},
\end{equation}
and the ratio $\Delta_s(x_j)/\Delta(x_j)$ is bounded by
\begin{equation}
    \begin{aligned}
    &\frac{2^{e+1}(1 - 2^{-p-1})}{\sqrt{2}}2^{1-p} \leq \Delta_s(x_j) \leq \frac{2^{e+1}(1+2^{-p})}{\sqrt{2}}2^{1-p}\\
    \implies & \frac{2^{e+1}(1 - 2^{-p-1})}{3 \cdot 2^{e-p}\sqrt{2}}2^{1-p} \leq \frac{\Delta_s(x_j)}{\Delta(x_j)} \leq \frac{2^{e+1}(1+2^{-p})}{3 \cdot 2^{e-p}\sqrt{2}}2^{1-p}\\
    \implies & \frac{2\sqrt{2}(1 - 2^{-p-1})}{3} \leq \frac{\Delta_s(x_j)}{\Delta(x_j)}  \leq \frac{2\sqrt{2}(1+2^{-p})}{3}.
    \end{aligned}
\end{equation}
When $p = 1$, the lower bound is equal to $1/\sqrt{2}$ and the upper bound is equal to $\sqrt{2}$. The lower bound is monotonically increasing in $p$ while the upper bound is monotonically decreasing in $p$, and both converge to $2\sqrt{2}/3$ as $p \rightarrow \infty$, which is between $1/\sqrt{2}$ and $\sqrt{2}$. This means both bounds lie within $[1/\sqrt{2},\, \sqrt{2}]$, so for within exponent block bins and exponent boundary bins,
\begin{equation}\label{eq:ratio_bounds_appendix}
    -\frac{1}{2} \leq \log\left[\frac{\Delta_s(x_j)}{\Delta(x_j)}\right] < \frac{1}{2}.
\end{equation}

\paragraph{Bound on entropy difference from bin ratios}

Hence, for every $x_j\in \mathbb U\setminus (\mathbb{S} \cup \mathbb{O})$,
\begin{equation}
\abs{\log\bracket{\frac{\Delta_s(x_j)}{\Delta(x_j)}}} \leq \frac{1}{2}.
\end{equation}

The remaining points are exactly the overflow and underflow bins collected in $\mathbb{S} \cup \mathbb{O}$, where we do not claim a uniform pointwise bound and instead keep their contribution explicitly. By the triangle inequality and~\eqref{eq:ratio_bounds_appendix}:
\begin{equation}
\begin{aligned}
&\left|
\Ht(\Xv_{fp}) - \Ht_{s}(\Xv_{fp}) \right| 
= \Bigg | -\int_{\mathbb U^d} f_{\Xv}(\xv)\,\log\left[f_{\Xv}(\xv)\prod_{j=1}^d \Delta(x_j)\right]d\xv  
+ \int_{\mathbb R^d} f_{\Xv}(\xv)\,\log\left[f_{\Xv}(\xv)\prod_{j=1}^d \Delta_s(x_j)\right]d\xv \Bigg |\\
&= \Bigg | \sum_{j=1}^d \int_{\mathbb U\setminus\mathbb{S}}f_{X_j}(x_j) \,\log\abs{\frac{ \Delta_s(x_j)}{\Delta(x_j)}}dx_j 
+ \int_{\mathbb S^d \cup \mathbb O^d} f_{\Xv}(\xv)\,\log\left[f_{\Xv}(\xv)\prod_{j=1}^d \Delta_s(x_j)\right]d\xv \Bigg |\leq \frac{d}{2} + \epsilon.
\end{aligned}
\end{equation}
\end{proof}

\section{Computing the exact entropy of quantized random variables}\label{app:exact_entropy_float}

Here we compute the exact discrete entropy of quantized random variables, such as normally distributed random variables stored in floating-point numbers.

In this paper, we will use a simple normalized floating-point format with midpoint rounding that does not use subnormal numbers (as defined in \cite{muller2010handbookfpn}) and where zero is not in the representable set.

Let $p \in \N{}$, $E \in \N{} \cup \{0\}$, and $\alpha \in \R{}$. The following map is defined on $\R{}\setminus\{0\}$; since all distributions used in this paper are absolutely continuous, the event that the value zero is realized has probability zero. Define \begin{equation}
    \mathrm{round}_{p}(\alpha) \triangleq \paren{\argmin\limits_{i \in \{0, \dots, 2^{p-1}\}}\abs{\alpha - i2^{-(p-1)}}}2^{-(p-1)},
\end{equation} where ties in the $\argmin$ function are broken towards the higher index $i$. Assuming rounding to the nearest representable value, the value stored on a floating-point number can be written as
\begin{equation}
    x_{fp}(x) \triangleq s_{fp}(x) \times 2^{e_{fp}(x)} \times \paren{1 + m_{fp}(x)},
\end{equation}
where\footnote{$\tilde e_{fp}(x)$ must be introduced to handle cases where the mantissa value rounds up to the next exponent level. This occurs when $\mathrm{round}
      _{p}\paren{|x|2^{-\tilde e_{fp}(x)} - 1} = 1$.} \begin{equation}
\label{eq:sem_fp}
    \begin{aligned}
        s_{fp}(x) &\triangleq (-1)^{\indic{x < 0}},\\
        \tilde e_{fp}(x) &\triangleq -\paren{2^{E-1} - 1}\indic{\log|x| < -\paren{2^{E-1} - 1}}\\
        &+ \sum\limits_{i=0}^{2^E - 1}\bracket{i - \paren{2^{E-1} - 1}}\indic{i - \paren{2^{E-1} - 1} \leq \log|x| < i+1 - \paren{2^{E-1} - 1}} + 2^{E-1}\indic{\log|x| \geq 2^{E-1} + 1}\\
          m_{fp}(x) &\triangleq \mathrm{round}_{p}\paren{|x|2^{-\tilde e_{fp}(x)} - 1}\indic{\log|x| < 2^{E-1}+1}\indic{\mathrm{round}_{p}\paren{|x|2^{-\tilde e_{fp}(x)} - 1
     } \neq 1} \\
          &+ \paren{1 - 2^{-(p-1)} }\paren{\indic{\log|x| \geq 2^{E-1}+1} + \indic{2^{E-1} \leq \log|x| < 2^{E-1}+1}\indic{\mathrm{round}_{p}\paren{|x|2^{-\tilde e_{fp}(x)}
     - 1} = 1}}\\e_{fp}(x) &\triangleq \min\paren{2^{E - 1}, \tilde e_{fp}(x) + \indic{\mathrm{round}_{p}\paren{|x|2^{-\tilde e_{fp}(x)} - 1} = 1}}.
       \end{aligned}
   \end{equation}
The quantities in Eq.~\eqref{eq:sem_fp} are decoded numerical values: $s_{fp}(x)$ is a single sign bit, $e_{fp}(x)$ is the binary exponent for $E \geq 1$ (for $E=0$, the format has a single exponent level with implicit exponent value $1/2$), and $m_{fp}$ is the significand (or mantissa) which encodes the binary significant digits of $x$ in the form $1.m_{fp}(x)$; in the physical registers these are encoded as binary integers using $1$, $E$, and $(p-1)$ bits respectively \cite{muller2010handbookfpn, Goldberg1991_FPN_Intro}. Since each stored field is a finite set of non-negative integers constrained by the number of bits in the register, this structure can only represent a finite set $U_{fp}$ of representable numbers along the real number line. For example, the standard single-precision IEEE-754 format uses one bit for $s_{fp}$, 23 bits for $m_{fp}$, and eight bits for $e_{fp}$. Note that the format defined by Eq.~\eqref{eq:sem_fp} is an idealized normalized floating-point format: it does not include zero, subnormal numbers, infinities, or NaNs, and its exponent range ($e_{\min} = -(2^{E-1}-1)$, $e_{\max} = 2^{E-1}$) differs slightly from the IEEE-754 standard. 
The structure of a floating-point number is illustrated in Fig.~\ref{fig:fp_diag}. See \cite{muller2010handbookfpn, Goldberg1991_FPN_Intro} for further details on the floating-point representation.

\begin{theorem}[Entropy of a Clipped and Arbitrarily Midpoint Quantized Random Variable]\label{thm:quantized_ent}
Let $X$ be an absolutely continuous random variable with cumulative distribution function $F$. Assume there are $K$ representable values in the quantization scheme and denote the set of these values as $\{u_1, \dots, u_K\}$, where $u_1 < \dots < u_K$. If $X_Q$ is the resulting clipped and midpoint quantized representation of $X$, then the discrete entropy of $X_Q$ is
\begin{align}
    \begin{aligned}
        H(X_Q) &= -F\paren{\frac{u_1 + u_2}{2}}\log\bracket{F\paren{\frac{u_1 + u_2}{2}}}\\
        &- \sum_{i = 2}^{K-1}\bracket{F\paren{\frac{u_{i+1} + u_i}{2}} - F\paren{\frac{u_{i} + u_{i-1}}{2}}}\log\bracket{F\paren{\frac{u_{i+1} + u_i}{2}} - F\paren{\frac{u_{i} + u_{i-1}}{2}}}\\
        &-\bracket{1 - F\paren{\frac{u_{K-1} + u_K}{2}}}\log\bracket{1 - F\paren{\frac{u_{K-1} + u_K}{2}}}.
        \end{aligned}
\end{align}
when $K \geq 3$. When $K = 2$, $H(X_Q) = -F\paren{\frac{u_1 + u_2}{2}}\log\bracket{F\paren{\frac{u_1 + u_2}{2}}} - \bracket{1 -  F\paren{\frac{u_1 + u_2}{2}}}\log\bracket{1 -  F\paren{\frac{u_1 + u_2}{2}}}$, and when $K = 1$, $H(X_Q) = 0$.

\end{theorem}

\begin{proof}
Since $X_Q$ is a clipped and midpoint quantized representation of $X$, we have
\begin{equation}
    X_Q = \begin{cases}
    \sum\limits_{i = 2}^{K-1} u_i\indic{X \in \left[\frac{u_{i} + u_{i-1}}{2},\frac{u_{i+1} + u_i}{2}\right)} + u_1\indic{X < \frac{u_1 + u_2}{2}} + u_K\indic{X \geq \frac{u_{K-1} + u_K}{2}} \text{ if } K \geq 3\\
    u_1\indic{X < \frac{u_1 + u_2}{2}} + u_2\indic{X \geq \frac{u_1 + u_2}{2}} \text{ if } K = 2\\
    u_1 \text{ if } K = 1.
    \end{cases}
\end{equation}

For $K = 1$, we see that $P\{X_Q = u_1\} = 1$ so $H(X_Q) = 0$. For $K = 2$, $P\{X_Q = u_1\} = F\paren{\frac{u_1 + u_2}{2}} = 1 - P\{X_Q = u_2\}$, so 
\begin{equation}
H(X_Q) = -F\paren{\frac{u_1 + u_2}{2}}\log\bracket{F\paren{\frac{u_1 + u_2}{2}}} - \bracket{1 -  F\paren{\frac{u_1 + u_2}{2}}}\log\bracket{1 -  F\paren{\frac{u_1 + u_2}{2}}}.
\end{equation}

Finally, for $K \geq 3$,

\begin{equation}
\begin{aligned}
    P\{X_Q = u_i\} &=  \begin{cases} \int\limits_{\frac{u_{i} + u_{i-1}}{2}}^{\frac{u_{i+1} + u_i}{2}} f_X(x)dx \text{ if } i \in [2, K-1]\\
    \int\limits_{-\infty}^{\frac{u_1 + u_2}{2}} f_X(x)dx \text{ if }  i = 1\\
    \int\limits_{\frac{u_{K-1} + u_K}{2}}^{\infty} f_X(x)dx \text{ if }  i = K\\ 
    \end{cases}\\
    &=\begin{cases} F\paren{\frac{u_{i+1} + u_i}{2}} - F\paren{\frac{u_{i} + u_{i-1}}{2}} \text{ if } i \in [2, K-1]\\
    F\paren{\frac{u_1 + u_2}{2}}  \text{ if } i = 1\\
    1 - F\paren{\frac{u_{K-1} + u_K}{2}} \text{ if } i = K
    \end{cases}.
\end{aligned}
\end{equation}

Inserting this into the formula for discrete entropy allows us to derive the following
\begin{equation}
    \begin{aligned}
        H(X_Q) & = - \sum_{i = 1}^K P\{X_Q = u_i\} \log\bracket{P\{X_Q = u_i\}} \\
        & = -F\paren{\frac{u_1 + u_2}{2}}\log\bracket{F\paren{\frac{u_1 + u_2}{2}}}\\
        &- \sum_{i = 2}^{K-1}\bracket{F\paren{\frac{u_{i+1} + u_i}{2}} - F\paren{\frac{u_{i} + u_{i-1}}{2}}}\log\bracket{F\paren{\frac{u_{i+1} + u_i}{2}} - F\paren{\frac{u_{i} + u_{i-1}}{2}}}\\
        &-\bracket{1 - F\paren{\frac{u_{K-1} + u_K}{2}}}\log\bracket{1 - F\paren{\frac{u_{K-1} + u_K}{2}}}.
        \end{aligned}
\end{equation}
    
\end{proof}

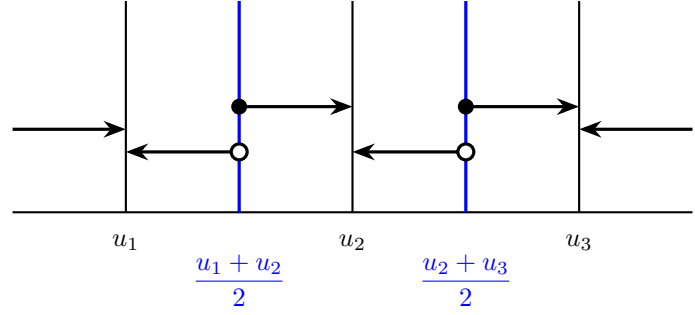
\begin{figure}[htbp]
\centering
\begin{tikzpicture}[>=Stealth, thick]

\draw[black, thick] (-3.5, 0) -- (5.5, 0);

\draw[black, thick] (-2, 2.8) -- (-2, 0);
  \draw[black, thick] ( 1, 2.8) -- ( 1, 0);
  \draw[black, thick] ( 4, 2.8) -- ( 4, 0);

\draw[blue, very thick] (-0.5, 2.8) -- (-0.5, 0);
  \draw[blue, very thick] ( 2.5, 2.8) -- ( 2.5, 0);

\draw[black, -{Stealth}, very thick] (-0.5, 1.4) -- ( 1.0, 1.4);
  \fill[black] (-0.5, 1.4) circle (3pt);
  \draw[black, -{Stealth}, very thick] ( 2.5, 1.4) -- ( 4.0, 1.4);
  \fill[black] ( 2.5, 1.4) circle (3pt);

\draw[black, -{Stealth}, very thick] (-0.5, 0.8) -- (-2.0, 0.8);
  \fill[white] (-0.5, 0.8) circle (3pt);
  \draw[black, very thick] (-0.5, 0.8) circle (3pt);
  \draw[black, -{Stealth}, very thick] ( 2.5, 0.8) -- ( 1.0, 0.8);
  \fill[white] ( 2.5, 0.8) circle (3pt);
  \draw[black, very thick] ( 2.5, 0.8) circle (3pt);

\draw[black, -{Stealth}, very thick] (-3.5, 1.1) -- (-2.0, 1.1);
  \draw[black, -{Stealth}, very thick] ( 5.5, 1.1) -- ( 4.0, 1.1);

\node at (-2, -0.4) {$u_1$};
  \node at ( 1, -0.4) {$u_2$};
  \node at ( 4, -0.4) {$u_3$};
  \node[blue] at (-0.5, -0.9) {$\dfrac{u_1+u_2}{2}$};
  \node[blue] at ( 2.5, -0.9) {$\dfrac{u_2+u_3}{2}$};

\end{tikzpicture}

\caption{Clipping and midpoint quantization with $K=3$ representable values $\{u_1, u_2,u_3\}$. The blue vertical lines represent the midpoints, and the arrows depict the regions of the real line that map to each representable value at a black vertical line.}\label{fig:midpoint_quantization}
\end{figure}

\begin{corollary}[Entropy of a Random Variable Quantized to a Floating-point Number]\label{cor:gauss_fp_ent}
Let $X$ be an absolutely continuous random variable with cumulative distribution function $F$. Let $X_{fp}$ be the clipped and midpoint quantized floating-point representation of $X$ with an $E$-bit exponent and a $(p-1)$-bit significand. The discrete entropy of $X_{fp}$ is given by Theorem~\ref{thm:quantized_ent} where $K = 2^{E + p}$, 
\begin{align}
    u_i = \begin{cases}
    u^\prime_{i -2^{E + p - 1}}  \text{ if } i > 2^{E + p - 1}\\
    -u^\prime_{2^{E + p - 1} - (i-1)} \text{ otherwise},
\end{cases}
\end{align}
and 
\begin{align}
    u^\prime_i = 2^{\bracket{\left\lfloor (i-1)2^{-(p-1)}\right\rfloor - \paren{2^{E-1} - 1}}}\paren{1 + \bracket{(i-1)2^{-(p-1)} \bmod 1}}.
\end{align}
\end{corollary}

\begin{proof}
We will explicitly construct the ordered sequence of representable values $\{u_1, u_2, \dots, u_K\}$ from the structure of the floating-point format given in Eq.~\eqref{eq:sem_fp}, then apply Theorem~\ref{thm:quantized_ent}. There are $2^{p-1}$ distinct mantissa values for each exponent and $2^E$ exponent values, giving $2^{E+p-1}$ positive floating-point values in total. Including the negative values by symmetry, there are $K = 2^{E+p}$ representable values altogether.
The smallest value of the exponent is $e_{\min} = -(2^{E-1}-1)$, while the largest value is $e_{\max} = 2^{E-1}$.

We first enumerate the $2^{E+p-1}$ positive floating-point values $u^\prime_i$ in increasing order, indexed by $i = 1, 2, \ldots, 2^{E+p-1}$. For each fixed exponent value, we exhaust all $2^{p-1}$ mantissa values before incrementing the exponent.
Concretely, for index $i$, the exponent index is $\lfloor (i-1)/2^{p-1} \rfloor$, which starts at zero and steps up by one exactly every $2^{p-1}$ values of $i$. The mantissa index within that exponent block is $(i-1) \bmod 2^{p-1}$, which cycles through $0, 1, \ldots, 2^{p-1}-1$ repeatedly. Translating the exponent index into the true exponent by adding $e_{\min}$, and the mantissa index into its fractional value by multiplying by $2^{-(p-1)}$, gives
\begin{align}
    u^\prime_i = 2^{\left\lfloor (i-1)2^{-(p-1)}\right\rfloor - (2^{E-1}-1)}\Bigl(1 + \bracket{(i-1)2^{-(p-1)} \bmod 1}\Bigr).
\end{align}

The full sequence $u_1 < u_2 < \cdots < u_K$ must enumerate the negative and positive floating-point values in increasing order. Since the negative floating-point values are the mirror image of the positive ones, the most negative value corresponds to $-u^\prime_{2^{E+p-1}}$ and the least negative to $-u^\prime_1$. Therefore, for $i \leq 2^{E+p-1}$, we set
\begin{align}
    u_i = -u^\prime_{2^{E+p-1} - (i-1)},
\end{align}
which enumerates the negative values in increasing order. For $i > 2^{E+p-1}$, we set
\begin{align}
    u_i = u^\prime_{i - 2^{E+p-1}},
\end{align}
which enumerates the positive values in increasing order. With these $u_i$ and $K = 2^{E+p}$, Theorem~\ref{thm:quantized_ent} can be used to compute $H(\Xfp)$, completing the proof.
\end{proof}

\newpage
\section{Exact vs.\ Approximate Floating-point Entropy Figures}\label{app:exact_FPN_extrafigs}

\begin{figure}[h]
  \centering
  \begin{subfigure}{0.24\textwidth}
    \centering
    \includegraphics[width=\linewidth]{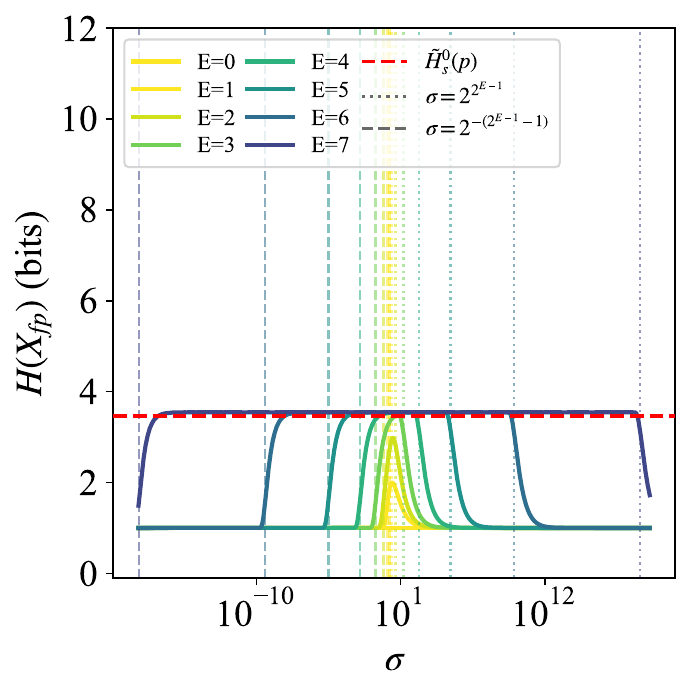}
    \caption{$p=1$.}
    \label{fig:app_sigma_p1}
  \end{subfigure}
  \begin{subfigure}{0.24\textwidth}
    \centering
    \includegraphics[width=\linewidth]{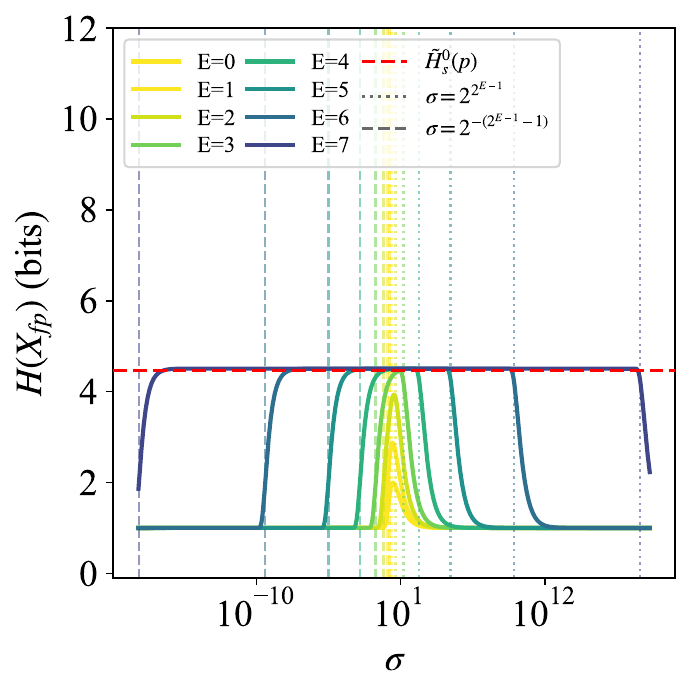}
    \caption{$p=2$.}
    \label{fig:app_sigma_p2}
  \end{subfigure}
  \begin{subfigure}{0.24\textwidth}
    \centering
    \includegraphics[width=\linewidth]{code/quantized_entropy_midpoint/sigma_p_E/entropy_sigma_midpoint_p3.pdf}
    \caption{$p=3$.}
    \label{fig:app_sigma_p3}
  \end{subfigure}
  \begin{subfigure}{0.24\textwidth}
    \centering
    \includegraphics[width=\linewidth]{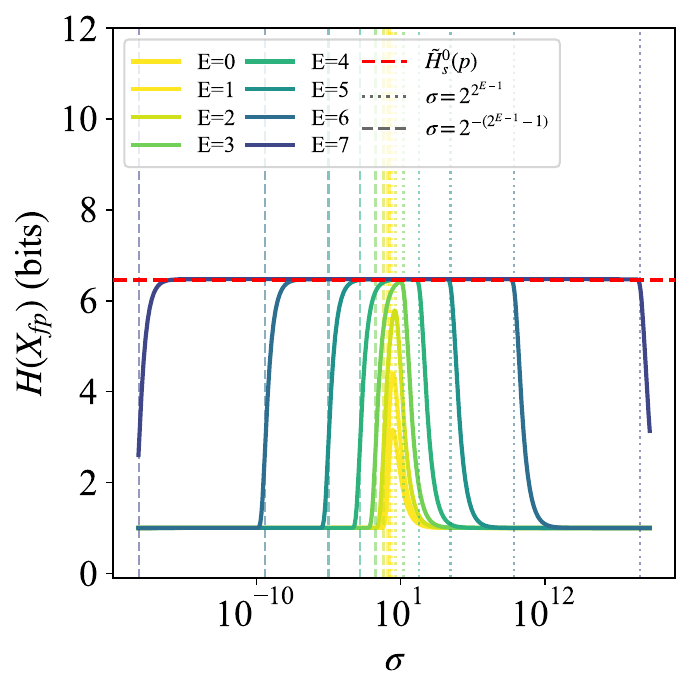}
    \caption{$p=4$.}
    \label{fig:app_sigma_p4}
  \end{subfigure}
  \begin{subfigure}{0.24\textwidth}
    \centering
    \includegraphics[width=\linewidth]{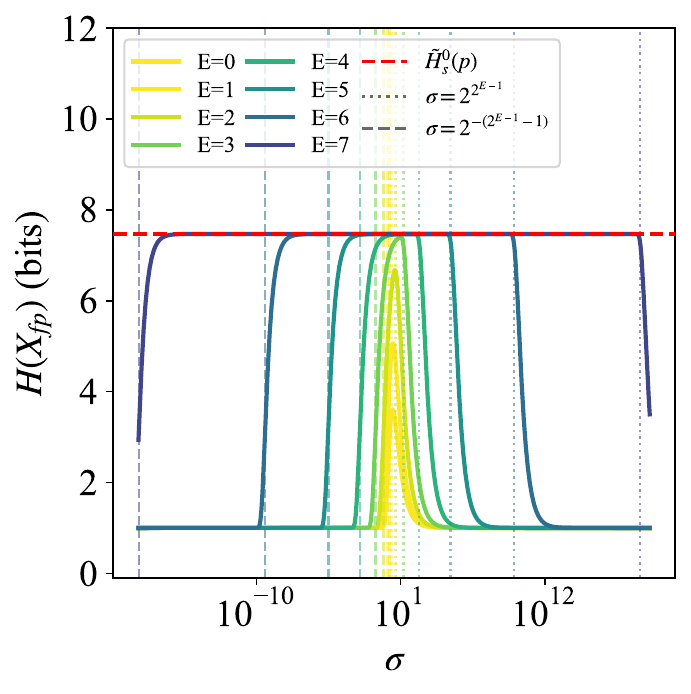}
    \caption{$p=5$.}
    \label{fig:app_sigma_p5}
  \end{subfigure}
  \begin{subfigure}{0.24\textwidth}
    \centering
    \includegraphics[width=\linewidth]{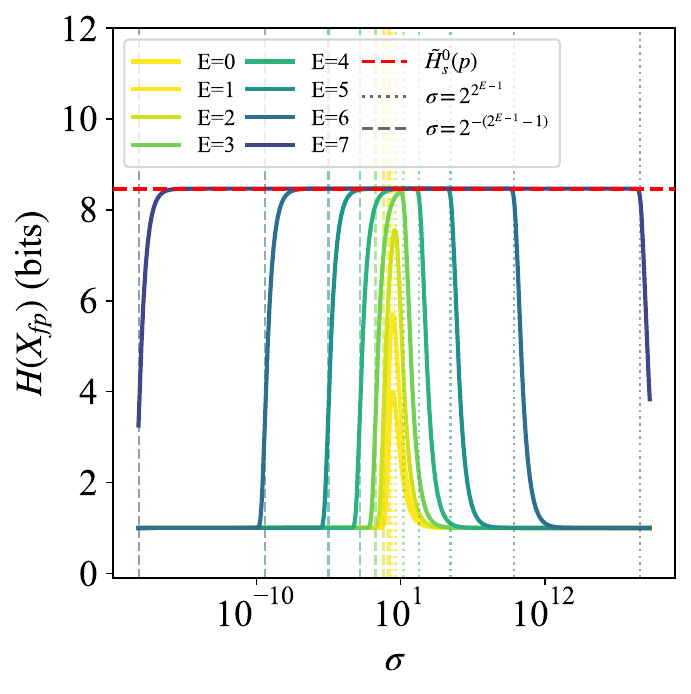}
    \caption{$p=6$.}
    \label{fig:app_sigma_p6}
  \end{subfigure}
  \begin{subfigure}{0.24\textwidth}
    \centering
    \includegraphics[width=\linewidth]{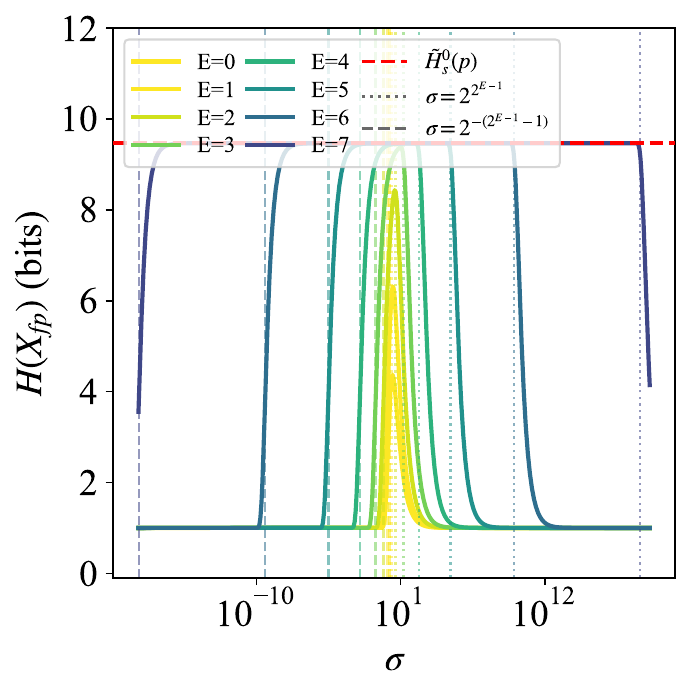}
    \caption{$p=7$.}
    \label{fig:app_sigma_p7}
  \end{subfigure}
  \begin{subfigure}{0.24\textwidth}
    \centering
    \includegraphics[width=\linewidth]{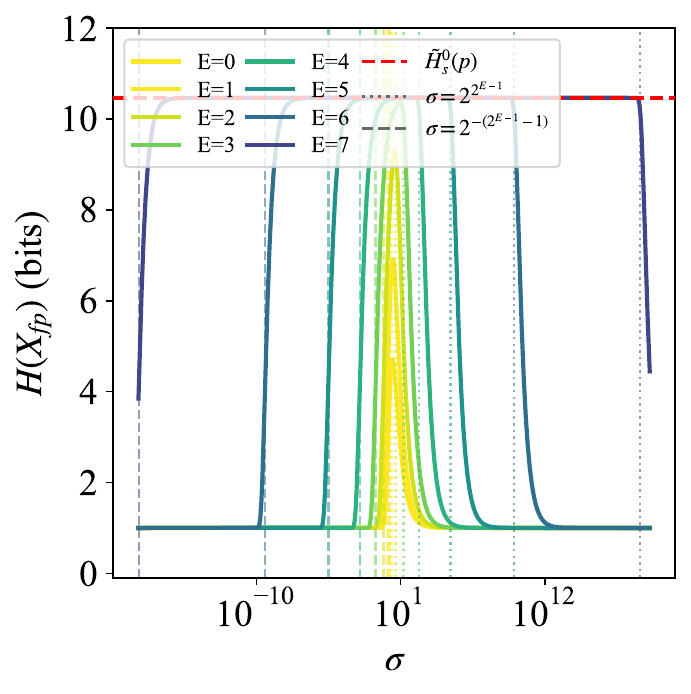}
    \caption{$p=8$.}
    \label{fig:app_sigma_p8}
  \end{subfigure}
  \caption{\textit{Exact midpoint-quantized entropy vs.\ standard deviation $\sigma$.} For each precision $p \in \{1,\ldots,8\}$, the exact discrete entropy $H(X_{fp})$ of $X \sim \gauss{0}{\sigma^2}$ (with $\mu=0$) is plotted as a function of $\sigma$ over a wide log-scale range. Each curve corresponds to a distinct value of exponent bits $E \in \{0,1,\ldots,7\}$. The vertical dashed lines mark $\sigma = 2^{e_{\min}}$ and the vertical dotted lines mark $\sigma = 2^{e_{\max}}$ for each $E$, and the horizontal red line shows the approximation. These plots were generated by sweeping $\sigma$ over 500 log-spaced points and computing the exact entropy via Corollary~\ref{cor:gauss_fp_ent}.}
  \label{fig:app_entropy_vs_sigma_all}
\end{figure}

\end{document}